\documentclass[hyperfootnotes=false,letterpaper]{article}

\usepackage[english]{babel}
\usepackage{cite}

\usepackage[utf8]{inputenc}
\usepackage{fullpage}
\usepackage{amsmath, amsfonts, amssymb, amsthm}
\usepackage{color}
\usepackage{xcolor}
\definecolor{darkgreen}{rgb}{0,0.5,0}
\definecolor{darkblue}{rgb}{0,0,0.6}

\usepackage[linesnumbered,noresetcount,vlined,ruled,procnumbered]{algorithm2e}
\usepackage[pagebackref=true]{hyperref}
\hypersetup{
  unicode=false,          %
  colorlinks=true,        %
  linkcolor=darkblue,          %
  citecolor=darkgreen,        %
  filecolor=magenta,      %
  urlcolor=blue
}

\usepackage[capitalize, nameinlink]{cleveref}

\usepackage{graphicx}
\usepackage{comment}
\usepackage{thm-restate}
\usepackage{bbm}
\usepackage{tikz}
\usepackage[linewidth=1pt]{mdframed}
\usepackage[most]{tcolorbox}
\definecolor{block-gray}{gray}{0.85}
\newtcolorbox{shadequote}{colback=block-gray,grow to right by=-2mm,grow to left by=-2mm,
  boxrule=0pt,boxsep=0pt,breakable}
\usepackage{color}
\usepackage[inline]{enumitem}

\usepackage{caption}
\usepackage{cancel}
\usepackage{subcaption}
\usepackage{mathrsfs}
\usepackage{xspace}

\usepackage{mathtools} %
\usepackage{booktabs} %
\usepackage{makecell}%
\usepackage{tabularx}

\SetKwComment{Comment}{$\quad\triangleright$\ }{}
\newcommand{\nosemic}{\renewcommand{\@endalgocfline}{\relax}}%

\usepackage[normalem]{ulem} %

\crefalias{AlgoLine}{line}%

\makeatletter
\let\cref@old@stepcounter\stepcounter
\def\stepcounter#1{%
  \cref@old@stepcounter{#1}%
  \cref@constructprefix{#1}{\cref@result}%
  \@ifundefined{cref@#1@alias}%
  {\def\@tempa{#1}}%
  {\def\@tempa{\csname cref@#1@alias\endcsname}}%
  \protected@edef\cref@currentlabel{%
    [\@tempa][\arabic{#1}][\cref@result]%
    \csname p@#1\endcsname\csname the#1\endcsname}}
\makeatother

\usepackage[titletoc,title]{appendix}

\newif\ifnotes
\notestrue
\ifnotes
    \newcommand{\keren}[1]{{\ifnotes \color{blue}{Keren: #1}
                \fi}}
    \newcommand{\tomer}[1]{{\ifnotes \color{brown}{Tomer: #1}
                \fi}}
    \newcommand{\virgi}[1]{{\ifnotes \color{red}{Virgi: #1}
                \fi}}
\else
    \newcommand{\keren}[1]{}
    \newcommand{\tomer}[1]{}
    \newcommand{\virgi}[1]{}
\fi

\newtheorem{theorem}{Theorem}[section]
\newtheorem{lemma}[theorem]{Lemma}
\newtheorem{remark}{Remark}
\newtheorem{definition}{Definition}

\newtheorem{proposition}{Proposition}
\newtheorem{claim}{Claim}
\newtheorem{corollary}{Corollary}
\newtheorem{conclusion}{Conclusion}

\newtheorem{hypothesis}{Hypothesis}

\makeatletter
\renewenvironment{proof}[1][\proofname]{\par
    \pushQED{\qed}%
    \normalfont \topsep6\p@\@plus6\p@\relax
    \trivlist
    \item\relax
    {\bfseries\boldmath
        #1\@addpunct{.}}\hspace\labelsep\ignorespaces
}{%
    \popQED\endtrivlist\@endpefalse
}
\makeatother

\newcommand{\apm}{(1\pm\eps)}
\newcommand{\apmp}{(1\pm\eps')}

\newcommand{\MM}[1]{\mathsf{MM}
(#1)}

\newcommand{\brak}[1]{\left(#1\right)}
\newcommand{\sbrak}[1]{\left[#1\right]}

\newcommand{\Exp}[1]{\mathbb{E}\left[ #1 \right]}
\newcommand{\Var}[1]{\mathsf{Var}\left[ #1 \right]}
\renewcommand{\Pr}[1]{{\mathrm{Pr}}\left[ #1 \right]}
\newcommand{\set}[1]{\left\{ #1 \right\}}
\newcommand{\sset}[1]{\{#1\}}

\newcommand{\polylog}[1]{\mathrm{polylog}\brak{#1}}

\newcommand{\ch}[2]{{#1 \choose #2}}
\newcommand{\Bin}[2]{\mathsf{Bin}\brak{{#1},{#2}}}

\newcommand{\BO}[1]{\mathcal{O}
(#1)}
\newcommand{\BC}{\mathcal{O}\brak{1}}
\newcommand{\TO}[1]{\tilde{\mathcal{O}}
(#1)}
\newcommand{\Omc}[1][1]{\Omega\brak{#1}}
\newcommand{\zrn}[1]{\set{0,1,\ldots,#1}}
\newcommand{\zrnone}[1]{\set{1,2,\ldots,#1}}

\newcommand{\eps}{\varepsilon}
\newcommand{\Em}[1]{\mathcal{E}_{#1}}

\newcommand{\HH}{\mathcal{H}}

\newcommand{\XC}{\mathcal{X}}

\newcommand{\AC}{\mathcal{A}}

\newcommand{\FF}{\mathcal{F}}
\newcommand{\FFall}{\mathcal{F}_{\mathrm{all}}}

\newcommand{\Ebb}{\mathcal{E}_{\mathrm{Find\mhyphen Heavy}}}

\newcommand{\FH}{\mathbb{H}}
\newcommand{\FHs}{\mathcal{C}}

\newcommand{\pnp}{1-\frac{1}{n^3}}
\newcommand{\fpr}[2]{1-\frac{{#1}}{n^{#2}}}

\DeclareMathSymbol{\mhyphen}{\mathord}{AMSa}{"39}

\makeatletter
\newcommand{\whp}{%
  w.h.p.\@ifnextchar.{\@gobble}{\xspace}%
}
\makeatother

\DeclarePairedDelimiter\abs{\lvert}{\rvert}%

\DeclarePairedDelimiter{\ceil}{\lceil}{\rceil}

\makeatletter
\let\oldabs\abs
\def\abs{\@ifstar{\oldabs}{\oldabs*}}

\crefname{claim}{Claim}{Claims} %
\crefname{ineq}{inequality}{inequalities} %
\crefname{proof}{Proof}{Proofs} %
\crefname{line}{Line}{Lines}
\crefname{algorithm}{Algorithm}{Algorithms}
\crefname{black box}{Black Box}{Black Boxes}

\newcommand{\DD}[1][\Lambda]{\hat{\mathsf{D}}\brak{#1}}
\newcommand{\DDF}{\hat{\mathsf{D}}(\Lambda p^h)}
\newcommand{\K}{\mathsf{K}}

\newcommand{\Thint}{{\mathsf{W}}}
\newcommand{\W}{{\mathsf{W}}}

\newcommand{\haT}{\hat{t}}
\newcommand{\htt}{\hat{t}}

\newcommand{\algCore}[1]{\mathsf{Template}_{\eps'}\brak{#1}}
\newcommand{\algCoreNP}{\mathsf{Template}}

\newcommand{\algWrap}[1]{\mathsf{Doubling\mhyphen Template}\brak{#1}}
\newcommand{\algWrapNP}{\mathsf{Doubling\mhyphen Template}}

\newcommand{\Cialg}{\mathsf{Count} \mhyphen\mathsf{Heavy}}
\newcommand{\FindHalg}{\mathsf{Find} \mhyphen\mathsf{Heavy}}

\newcommand{\neweps}{\frac{\eps}{4\log n}}

\newcommand{\Qb}{{8\log^{4}(n)}}
\newcommand{\BoundLam}[1]{\tau_{#1}\cdot \frac{\eps^2}{8Q}}

\newcommand{\rr}{400\log n}

\newcommand{\czlog}{{\brak{\log n}^{h^2}}}
\newcommand{\clog}{\czlog}

\newcommand{\Prod}{\mathsf{Product}_h}
\newcommand{\Prd}{\mathsf{Product}_h(\Lambda)}
\newcommand{\Prodd}{\mathsf{Product}}

\newcommand{\eo}{1-1/e}

\newcommand{\nemp}{\neq\emptyset}

\newcommand{\kVal}{\log^4 n}

\newcommand{\threeclass}{\varphi_{1},\varphi_{2},\varphi_{h}}
\newcounter{blackbox}[section] %
\newenvironment{blackbox}[1][]{%
  \refstepcounter{blackbox}%
  \begin{mdframed}[
    frametitle=#1,
    frametitlealignment=\centering,
    backgroundcolor=gray!10,
    linewidth=1pt,
    innerleftmargin=8pt,
    innerrightmargin=8pt,
    innertopmargin=10pt,
    innerbottommargin=10pt,
    splittopskip=\topskip,
    skipabove=\topsep,
    skipbelow=\topsep,
    userdefinedwidth=0.95\textwidth,
    align=center
  ]%
}{%
  \end{mdframed}%
}

\crefname{blackbox}{Black Box}{Black Boxes}

\newcommand{\omegaval}{2.371552}
\newcommand{\alphaval}{0.321334}

\makeatletter
\renewcommand\appendix{\par
  \setcounter{section}{0}%
  \setcounter{subsection}{0}%
  \gdef\thesection{\@arabic\c@section}}
\makeatother

\author{Keren Censor-Hillel \thanks{Department of Computer Science, Technion. \texttt{ckeren@cs.technion.ac.il}. The research is supported in part by the Israel Science Foundation (grant 529/23).} 
\and
Tomer Even \thanks{Department of Computer Science, Technion. \texttt{tomer.even@campus.technion.ac.il}.} 
\and 
Virginia Vassilevska Williams \thanks{Massachusetts Institute of Technology, Cambridge, MA, USA. \texttt{virgi@mit.edu}. Supported by NSF Grant CCF-2330048, BSF Grant 2020356, and a Simons Investigator Award.}}

\begin{document}
\title{Fast Approximate Counting of Cycles}
\date{}

\maketitle

\begin{abstract}
  We consider the problem of approximate counting of triangles and longer fixed length cycles in directed graphs. For triangles,  T\v{e}tek [ICALP'22] gave an algorithm that returns a $(1\pm\eps)$-approximation in $\tilde{O}(n^\omega/t^{\omega-2})$ time, where $t$ is the unknown number of triangles in the given $n$ node graph and $\omega<2.372$ is the matrix multiplication exponent. We obtain an improved algorithm whose running time is, within polylogarithmic factors the same as that
  for multiplying an $n\times n/t$ matrix by an $n/t \times n$ matrix. We then extend our framework to obtain the first nontrivial $(1\pm \eps)$-approximation algorithms for the number of $h$-cycles in a graph, for any constant $h\geq 3$. Our running time is
  \[\tilde{O}(\MM{n,n/t^{1/(h-2)},n}), \textrm{the time to multiply } n\times \frac{n}{t^{1/(h-2)}} \textrm{ by } \frac{n}{t^{1/(h-2)}}\times n \textrm{ matrices}.\]

  Finally, we show that under popular fine-grained hypotheses, this running time is optimal.
\end{abstract}

\setcounter{tocdepth}{3}
\tableofcontents

\pagebreak
\section{Introduction}
Detecting small subgraph patterns inside a large graph is a fundamental computational task with many applications.
Research in this domain has flourished, leading to fast algorithms for many tractable versions of the subgraph isomorphism problem: given a fixed (constant size) graph $H$, detect whether a large graph $G$ contains $H$ as a subgraph, list all copies of $H$ in $G$, count the copies (exactly or approximately) and more.

The topic of this paper is the fast estimation of the number of copies of a pattern $H$ in a graph $G$. One of the most studied patterns $H$ is the {\em triangle} whose detection, listing and approximate counting has become a prime testing ground for ideas in classic graph algorithms \cite{itai1977finding,alon1997finding,nesetril,BjorklundPWZ14,Patrascu10}, sublinear and distributed algorithms \cite{GonenRS11,EdenRS18,eden2022sublinear,EdenRR22,eden2017approximately,censor2022deterministic,czumaj2020detecting,izumi2017triangle,dolev2012tri,tvetek2022approximate,censor2020sparse,chang2021near,ChangS20}, streaming \cite{buriol2006counting,BecchettiBCG10,Bar-YossefKS02,JowhariG05,KaneMSS12}, parallel \cite{BiswasELRM22,KoldaPPST14,TangwongsanPT13} algorithms and more. 
This is largely because triangles are arguably the simplest subgraph patterns and moreover, often algorithms for the triangle version of the (detection, counting or listing) problem formally lead to algorithms for other patterns as well (see Ne\v{s}etril and Poljak \cite{nesetril}).

Detecting and finding a triangle, and counting the number of triangles in an $n$-vertex graph can all be reduced to fast matrix multiplication \cite{itai1977finding}, and the fastest algorithm for these versions has running time $O(n^\omega)$, where $\omega$ is the exponent of square matrix multiplication, currently $\omega\leq \omegaval$ \cite{williams2023new}. It is also believed that even detecting a triangle requires $n^{\omega-o(1)}$ time, due to known fine-grained reductions that show that Boolean matrix multiplication and triangle detection are equivalent, at least for combinatorial algorithms \cite{focsy}.

Obtaining an approximate count $\hat{t}$ to the number of triangles $t$ in an $n$-node graph such that $(1-\eps)t\leq \hat{t}\leq (1+\eps)t$ for an arbitrarily small constant $\eps>0$ can be used to detect whether a graph has a triangle, as the algorithm would be able to distinguish between $t=0$ and $t=1$. Thus, it is plausible that when the number of triangles is $O(1)$, $n^{\omega-o(1)}$ time is needed to obtain an approximate triangle count. 

When the number of triangles $t$ in $G$ is large, however, a simpler sampling approach can obtain a good estimate of $t$: repeatedly sample a triple of vertices and check whether they form a triangle; in expectation $O(n^3/t)$ samples are sufficient to get a constant factor approximation.

The best known algorithm for approximately counting triangles is by T\v{e}tek~\cite{tvetek2022approximate}, with running time $\tilde{O}(n^\omega/t^{\omega-2})$. When $t$ becomes constant, the running time becomes $\tilde{O}(n^\omega)$, which is believed to be optimal, as we mentioned earlier. 
When $t$ becomes $\Theta(n)$, the running time is the same as the na\"ive sampling algorithm, $\tilde{O}(n^2)$. This quadratic running time is provably necessary even for randomized algorithms (see \cite{eden2017approximately}).
Nevertheless, it is unclear whether  $\tilde{O}(n^\omega/t^{\omega-2})$ time is needed for all values of $t$ 
between $\BO{\polylog{n}}$ and $\Omega(n/\polylog{n})$.

\begin{center}
{\em Is there a faster algorithm for approximate triangle counting when the triangle count is in $[\Omega(1),O(n)]$?}
\end{center}

As triangle counting is an important special case of fixed subgraph isomorphism counting, a natural question is, what is the fastest algorithm for approximately counting arbitrary subgraphs $H$? 

Dell, Lapinskas and Meeks \cite{DellLM22} provide a general reduction from approximate $H$ counting to detecting a ``colorful'' $H$ in an $n$-node, $m$-edge graph, so that a $T(n,m)$ time detection algorithm can be converted into an $\tilde{O}(\eps^{-2}T(n,m))$ time $(1\pm \eps)$-approximation algorithm. For many patterns \footnote{This equivalence is not true in general: detecting  cycles of fixed even length is believed to be computationally easier than detecting colorful even cycles which are known to be equivalent to the directed version of the problem.} such as triangles and $k$-cliques or directed $h$-cycles, the colorful and normal versions of the detection problems are equivalent (e.g. via color-coding \cite{alon1995color} and layering). While the detection running time is provably necessary to approximately count when the number of copies of the pattern is constant, similarly to the case of triangles, when the number of copies $t$ is large, faster sampling algorithms are possible.
Unfortunately, the reduction of \cite{DellLM22} doesn't seem easy to extend to provide runtime savings that grow with $t$. Thus, we ask:
\begin{center}
{\em What is the best approximate counting algorithm for subgraph patterns $H$ with running time depending on the number $t$ of copies of $H$?}
\end{center}

As triangles are also cycles, one special case of the above question is when $H$ is a cycle on $h$ vertices.

\subsection{Our Contribution}
The main result of our paper is a new algorithm for approximating the number of $h$-cycles in a given {\em directed} graph, for any constant $h\geq 3$, together with a conditional lower bound from a fine-grained hypothesis, showing that the running time of our algorithm is likely tight.

Our main theorem is:

\begin{restatable}[Approximating the Number of $h$-Cycles]{theorem}{ThmZero}\label{thm:main0}
    Let $G$ be a given graph with $n$ vertices and let $h\geq 3$ be a fixed integer.
    There is a randomized algorithm that outputs an approximation $\hat{t}$ for the number $t$ of $h$-cycles in $G$ such that  
    $\Pr{(1-\eps)t \leq \htt\leq (1+\eps)t}\geq 1-1/n^2$, for any constant $\eps>0$. %
    The running time is bounded by $\TO{\MM{n,n/t^{1/(h-2)},n}}$, the fastest running time to multiply an $n\times n/t^{1/(h-2)}$ matrix by an $n/t^{1/(h-2)}\times n$ matrix.
\end{restatable}
As long as $\omega>2$, the running time in the theorem is always upper-bounded by 
$$\TO{n^\omega/t^{\frac{\omega-2}{(h-2)(1-\alpha)}} +n^2},$$
where  $\omega\leq \omegaval$ is the square matrix multiplication exponent mentioned earlier and $\alpha\geq \alphaval$ \cite{williams2023new} is the largest real number such that one can multiply $n\times n^\alpha$ by $n^\alpha\times n$ matrices in $n^{2+o(1)}$ time.\footnote{We use $\MM{a,b,c}$ to denote the time complexity of multiplying two matrices with dimensions $a\times b$ and $b\times c$.}
It is easy to see that for any value $\omega>2$, our $\TO{n^\omega/t^{\frac{\omega-2}{1-\alpha}}}$ time for {\em triangles} is faster than the previous state of the art $\TO{n^\omega/t^{\omega-2}}$ for approximate counting of triangles \cite{tvetek2022approximate} for all $t$ between $\Omega(1)$ and $O(n)$, answering our first question in the introduction.
Figure~\ref{fig:compare} plots our two running times for approximate triangle counting together with T\v{e}tek's algorithm, naive sampling and the $O(n^\omega)$ time exact counting algorithm.

\begin{figure}[ht]
    \begin{minipage}{.48\textwidth}
        \centering
        \includegraphics[width=8cm]{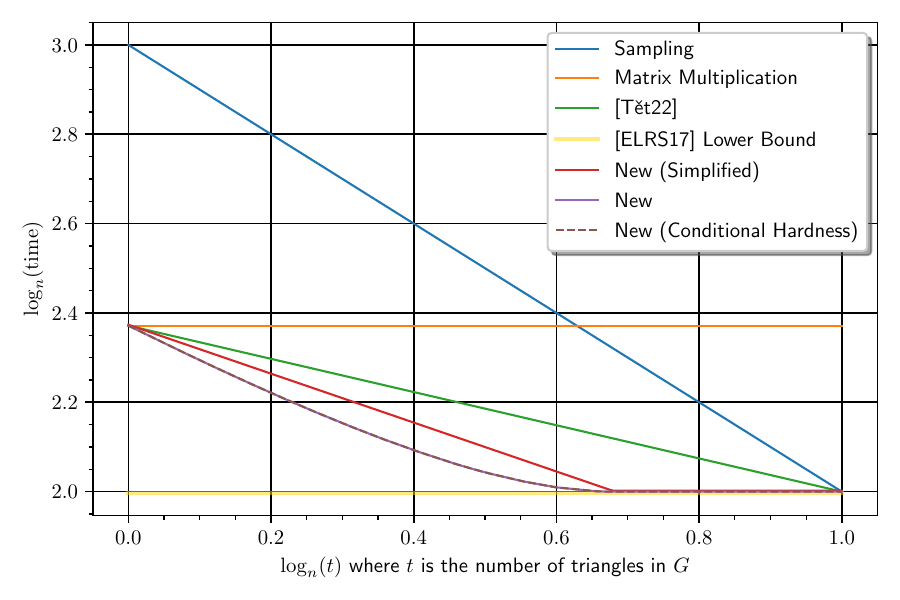}
        \caption{A comparison between our new running times for approximate triangle counting with prior work, together with the lower bounds, both conditional and unconditional.}
        \label{fig:compare}
    \end{minipage}%
    \begin{minipage}{.48\textwidth}
        \centering
        \begin{tabular}{ll}
            \toprule
            \textbf{Method}            & \textbf{Runtime}  \\
            \midrule
            Sampling      & $n^3/t$                         \\
            Matrix Multiplication           & $n^\omega$   \\
            \cite{tvetek2022approximate}  & $n^\omega/t^{\omega-2}$              \\
            \cite{eden2017approximately} Lower Bound & $n^2$ \\
            \midrule
            New (Simplified)   & $n^\omega/t^{(\omega-2)/(1-\alpha)}$ \\
            New   & $\mathsf{MM}(n,n,n/t)$ \\
            New (Conditional Hardness) & $\mathsf{MM}(n,n,n/t)$ \\
            \bottomrule
            \vspace{1cm}
        \end{tabular}
        \caption{Comparative Runtime Analysis}
        \label{table:runtime_analysis}
    \end{minipage}
\end{figure}

We obtain our algorithm via a simplification and generalization of T\v{e}tek's approach that allows us to both obtain an improved running time for triangles, but also to get faster algorithms for longer cycles. The approach can also extend to other patterns; we leave this as future work.

Our algorithm for longer cycles is arguably the first non-trivial algorithm for the problem with a negative dependence on the number of cycles $t$.
To our knowledge, prior to our work the only approximate counting algorithms for $h$-cycles for $h>3$ in directed graphs (or in undirected graphs when $h$ is odd\footnote{The detection problem for even $h$-cycles in undirected graphs is known to be much easier than that for odd cycles, and for directed graphs, as for every even constant integer $h$, an $O(n^2)$ time algorithm was developed by Yuster and Zwick \cite{evencycles}. Meanwhile, directed $h$-cycles and undirected odd $h$-cycles are believed to require $n^{\omega-o(1)}$ time to detect.}) were to either use naive random sampling resulting in an $O(n^4/t)$ time or to approximate the answer in the best $h$-cycle detection time, $O(n^\omega)$ (e.g. via \cite{DellLM22}), a running time that does not depend on $t$.

We complement our algorithms for approximately counting $h$-cycles with a {\em tight} conditional lower bound under a popular fine-grained hypothesis.
The $k$-Clique Hypothesis of fine-grained complexity (e.g. \cite{AbboudBW18,BackursT17,BringmannW17,dalirrooyfard2023listing})  postulates that the current fastest algorithms for detecting a $k$-clique in a graph (for constant $k\geq 3$) are optimal, up to $n^{o(1)}$ factors. We formulate a natural hypothesis about the complexity of {\em triangle detection} in unbalanced tripartite graphs that is motivated by and in part implied by the $k$-Clique Hypothesis. Then we show, under that hypothesis:

\begin{theorem}
Under fine-grained hypotheses, in the word-RAM model with $O(\log n)$ bit words, for any constant integer $h\geq 3$, any randomized algorithm that, when given an $n$ node directed graph $G$, can
    distinguish between $G$ being $h$-cycle-free and containing $\geq t$ $h$-cycles needs
    $\MM{n,n/{t}^{1/(h-2)},n}^{1-o(1)}$
    time. The same result holds for undirected graphs as well whenever $h$ is odd.
\end{theorem}

As any algorithm that can approximate the number $t$ of $h$-cycles multiplicatively, can distinguish between $0$ and $t$ $h$-cycles, we get that our algorithm running times are essentially tight.
We present our lower bound for triangles in Figure \ref{fig:compare} as a dotted line. Together with the lower bound by \cite{eden2017approximately}, our lower bound shows that our algorithm is (conditionally) optimal for all values of $t$.

Similarly to T\v{e}tek's algorithm, our algorithms for approximate $h$-cycle counting can be used to obtain improved $h$-cycle counting algorithms for {\em sparse} graphs, where the running time is measured in terms of the number of edges $m$. In particular, for triangles, one can simply substitute our new algorithm in terms of $n$ in T\v{e}tek's argument \cite{tvetek2022approximate} to obtain an approximate counting algorithm that runs in time $\TO{m^{2\omega/(\omega+1)}/t^{\frac{2(\omega-1)}{\omega+1}+\frac{\alpha(\omega-2)}{(1-\alpha)(\omega+1)}}}$. This running time is always faster than T\v{e}tek's $\TO{m^{2\omega/(\omega+1)}/t^{\frac{2(\omega-1)}{\omega+1}}}$ for any $\omega>2$.
One can similarly adapt the algorithms of Yuster and Zwick \cite{YusterZ04} and their analysis in \cite{DalirrooyfardVW21} to obtain approximate counting algorithms for longer cycles. We leave this to future work.

\subsection{Technical Overview}
To frame our technical contribution, we first briefly overview the approach of \cite{tvetek2022approximate}. The latter gives a randomized approximate counting algorithm for triangles, in time $\TO{n^{\omega}/t^{\omega - 2}}$.
In a nutshell, the algorithm finds a subset of vertices $S$ that contains all $\Lambda$-heavy vertices and no $\Lambda/\polylog{n}$-light vertices --- a vertex is called $\Lambda$-heavy if it participates in at least $\Lambda$ triangles, and otherwise it is called $\Lambda$-light. 
Then, the algorithm approximately counts the number of $\Lambda$-heavy triangles, which are triangles with at least one heavy vertex.
The algorithm then continues by sampling subsets of vertices from the set $V-S$, where each vertex is kept independently  uniformly at random with some probability, and processing  the sampled graphs by recursion.

Our technical contribution consists of three parts: (I) We simplify the above recursive approach, (II) we improve upon the component that finds heavy vertices, and (III) we improve upon the component that counts the number of cycles that contain heavy vertices. A compelling aspect of our technique is that it applies to any constant-length cycle.
In what follows, we overview each of these aspects.

\paragraph{I. The Recursive Template.}
In \cite{tvetek2022approximate}, each recursive invocation triggers seven further recursive calls and takes the median of their return values. As the depth of the recursion increases, the algorithm needs to use a precision parameter that becomes exponentially tighter.

In contrast, our algorithm initiates only a single recursive call. This allows us to avoid having to compute the median of several subcalls, which makes it easier
to apply standard amplification tools. 
In particular, it allows us to \emph{fix the precision parameter}. 

In addition, by reducing the recursive
tree to a ``path'', we \emph{simplify the analysis} of the running time.

A prime feature of our simplification of the recursion is that it allows us to present it as a \emph{template} for approximate subgraph counting for any fixed subgraph $\FH$, provided one designs the two black boxes (one for finding a superset $S$ of the $\Lambda$-heavy vertices with no $\Lambda/\polylog{n}$-light vertices, and another for approximately counting the number of copies of $\FH$ that intersect the set $S$).

For triangles, our improvement comes from simplifying the recursion and implementing the black box that finds heavy vertices faster, using rectangular matrix multiplication. Crucially, our implementations of these black boxes are general, in the sense that they apply to constant length $h$-cycle.
Specifically, we find the set $S$ in time $\TO{\MM{n,n/\Lambda^{1/(h-2)},n}}$. For such $S$, we find a
$(1+\epsilon)$ approximation for the number of copies of $\FH$ that intersect $S$ in time $\TO{n^2/\eps^3}$, which is independent of $\Lambda$ and of the cardinality of $S$.

\paragraph{II. Finding the Heavy Vertices.} 
The algorithm in \cite{tvetek2022approximate}
finds $\Lambda$-heavy vertices by sampling a subset of vertices uniformly at random, and using matrix multiplication to detect triangles inside the induced sampled subgraph. This takes $\TO{n^\omega/\Lambda^{\omega-2}}$ time, by $\TO{\Lambda^2}$ repetitions of multiplying $n/\Lambda \times n/\Lambda$ matrices.

At the heart of our approach for finding the heavy vertices lies non-uniform sampling, and computing the product of rectangular matrices rather than square ones. 
\sloppy{We obtain a running time of $\TO{\MM{n,n/\Lambda^{1/(h-2)},n}}=\TO{n^\omega/\Lambda^{\gamma_h}}$ for $h$-cycles, where $\gamma_h\triangleq\frac{\omega-2}{(1-\alpha)(h-2)}$. }
This comes from multiplying an $n\times n/\Lambda^{1/(h-2)}$ matrix by an $n/\Lambda^{1/(h-2)}\times n$ matrix.
In \cite{williams2023new} it was shown that $\alpha\geq\alphaval$, and therefore
$\frac{\omega-2}{1-\alpha}\geq 1.47(\omega-2)$, which establishes that our algorithm is never slower, and is faster (if $\omega>2$), where the gap increases with $\Lambda$ (for sufficiently large $\Lambda$, the folklore na\"ive sampling algorithm is superior).

Our starting point for finding the heavy vertices is the \emph{color-coding} technique of \cite{alon1995color}, which is widely employed for detecting $h$-cycles for $h=\BO{1}$.
This technique colors the vertices using $h$ colors uniformly independently at random and looks for \emph{colorful} $h$-cycles, which are $h$-cycles with exactly one vertex of each color. This restriction allows for faster detection but suffers some probability of missing $h$-cycles that are colored out of order, which can be overcome with sufficiently large probability by repeating this process.

To find colorful $h$-cycles, we utilize matrix multiplication. However, we do so in a refined manner. Rather than considering all vertices, we sample a subset of vertices from each color class in a nonuniform manner.
To illustrate this, consider the task of finding $\Lambda$-heavy vertices w.r.t. triangles.
We assign a random color to each vertex, and denote the three color classes by $V_1,V_2,V_3$.
We focus on identifying the $\Lambda$-heavy vertices within $V_1$.
Fix some $i\in[\log \Lambda]$.
We sample each vertex from $V_2$ with probability $2^i/\Lambda$, and we sample each vertex from $V_3$ with probability $1/2^i$. Let $H_i$ denote the induced graph obtained by all vertices from $V_1$ and the sampled vertices from $V_2$ and $V_3$, where we also direct edges from $V_j$ to $V_{j+1 \mod 3}$ and discard monochromatic edges.
We show that for every $\Lambda$-heavy vertex $v$, there exists an index $i\in[\log\Lambda]$ such that $v$ is in a triangle in $H_i$ with some probability at least $p_{\mathrm{heavy}}$, where $1/p_{\mathrm{heavy}}= \TO{1}$.
On the other hand, for $\Lambda/\polylog{n}$-light vertex $u$, we show that for every $i\in[\log \Lambda]$, the vertex $v$ is in a triangle in $H_i$ with probability at most $p_{\mathrm{heavy}}/2$. Therefore, we can distinguish between these cases.
Checking whether $v$ is in a triangle in $H_i$ can be done in $\TO{\MM{n,n,n/\Lambda}}$ time. 
Using amplification, we approximate the probability that $v$ is in a triangle in $H_i$ for every $v\in V_1$, and thus distinguish heavy vertices from lighter ones.

We generalize our approach for $h$-cycles by coloring the vertices with $h$ colors, and directing edges and discarding monochromatic edges, as for triangles. We also discard edges between non-consecutive color classes.
To find $\Lambda$-heavy vertices in the first color class, we sample in a nonuniform manner a subset of vertices from the $j$-th color class for $2\leq j\leq h$, where the product of the sampling probabilities of the color classes should be at most $1/\Lambda$, as for triangles. Let $H$ denote the obtained random induced subgraph.
The running time of computing the exact number of $h$-cycles each vertex in $H$ participates in, which is dominated by the size of the smallest color class in $H$, becomes $\TO{\MM{n,n,n/\Lambda^{1/(h-2)}}}$.
To see why, consider an $h$-partite graph $G$ with $n$ vertices in each part. Suppose $G$ has a vertex $v\in V_1$ with a neighbor $u\in V_2$, such that all $h$-cycles that intersect $v$, also intersect the edge $(u,v)$. Now, suppose each vertex set $V_j$, for $3\leq j\leq h$ has a subset $W_j$ of $\Lambda^{1/(h-2)}$ vertices, such that any $h$-tuple of the form $(v,u,w_3,w_4,\ldots, w_h)$ is an $h$-cycle in $G$, where $w_j\in W_j$ for $3\leq j\leq h$.
This implies that $v$ is $\Lambda$-heavy.
Note that if we keep each vertex from the $j$-th color class with a probability of $o(1/\Lambda^{1/(h-2)})$, we are unlikely to sample any vertex from $W_j$, and therefore we fail to learn that $v$ is $\Lambda$-heavy. On the other hand, if we sample vertices from each class with probability $\Omc[1/\Lambda^{1/(h-2)}]$, the smallest color class is of size $\Omc[n/\Lambda^{1/(h-2)}]$.

\paragraph{III. Counting the Heavy Copies.}

Given a graph $G$ and a subset of vertices $S$, where each vertex in $S$ participates in at least $a$ and at most $b$ copies of $h$-cycles for $h=\BC$, we show how to compute a $(1+\epsilon)$ approximation for the number of $h$-cycles that intersect the set $S$, in time $\TO{n^2 b/\eps a}$. In particular, the runtime is independent of size of the set $S$.

Consider a na\"ive approach, which approximates the average number of $h$-cycles that a vertex from $S$ intersects, and let us see why it fails to provide a good approximation for the total number of $h$-cycles intersecting $S$. Suppose $h=3$ and $\abs{S}=3$ and each vertex $v\in S$ participates in exactly one triangle in $G$.
Based solely on the number of triangles in which a vertex participates, it is impossible to distinguish the case where the set $S$ intersects one triangle in $G$ from the case in which it intersects three triangles in $G$. The issue here is double counting, as we did not avoid counting the same cycle more than once. For triangles, this obstacle can be avoided by replacing $G$ with a tripartite graph $G'$, where each of the three parts is a copy of $V$, and for each edge in $G$ there are six edges in $G'$, one for every ordered pair of parts. It is easy to see that every triangle in $G$ corresponds to six triangles in $G'$, and thus an estimate on $G'$ directly gives an estimate on $G$.
That is, we sample a subset $F$ of vertices from $S$, and for each copy $v'$ of $v\in F$ in the, say, first part of in the tri-partition $G'$,
we compute the number of triangles that go through it in $G'$. This avoids double counting, because each triangle in $G'$ intersects copies of the set $F$ from the first part at most once (as vertices in the same part form an independent set and hence cannot be in the same triangle). 
To summarize, restricting to triangles would significantly simplify this part to a single Chernoff inequality for independent random variables.
The source of this simplicity is that triangles are cliques. However, larger cycles are not cliques, and therefore we cannot apply this trick. 
If we simply repeated every vertex $h$ times to create a new graph $G'$ as in the triangle case above, we could create $h$-cycles in $G'$ that do not correspond to $h$-cycles in $G$: in fact, every {\em closed walk} of length $h$ would become an $h$-cycle. We use color-coding and a more careful analysis to overcome this issue.

First, we sample a subset of vertices from $S$ and for each sampled vertex $v$ we approximate the number of $h$-cycles which go through $v$ (and therefore intersect the set $S$). The crux of our algorithm is that in order to approximate the above, we approximate the number of $h$-cycles which go through $v$ and intersect the set $S$ exactly $k$ times, for each $1\leq k\leq h$.
The summation of these approximations yields our final result, and it naturally avoids the pitfall of double-counting.

To approximate the number of $h$-cycles intersecting the set $S$ exactly $k$ times for some $k$, we color the graph with $h-1$ colors and color vertex $v$ with the color $h$. This ensures that any colorful $h$-cycle intersects $v$.
Then, we choose $k-1$ color classes from the first $h-1$ classes, and retain only vertices of $S$ within those classes. For the remaining color classes, we keep only vertices that are not in $S$. The color class $h$ is fixed and always contains only $v$. This promises that each colorful $h$-cycle with $v$ in this auxiliary graph intersects $S$ exactly $k$ times.
The number of ways to choose exactly $k-1$ color classes that keep only vertices from $S$ is $\ch{h-1}{k-1}$.
We compute the number of $h$-cycles in each such auxiliary graph.
We prove that the expectation of this number is some fixed constant multiplicative factor off from the number of $h$-cycles intersecting $v$ and $S$ exactly $k$ times.
Finally, we prove that the variance of this random variable is suitably bounded. Therefore, conducting this process $\TO{b/(a\eps)}$ times enables us to obtain an $\apm$ approximation for its expectation 
by Chebyshev's inequality.
We compute the number of $h$-cycles in this auxiliary graph using rectangular matrix multiplication. Since the auxiliary graph is $h$-partite and one part contains only a single vertex, we get a running time of ${\MM{n,n,1}}=\TO{n^2}$. Thus, we achieve our claimed running time of $\TO{n^2b/\eps a}$. 

We mention that we invoke this procedure on the set of vertices given by the previous component of finding heavy vertices, which is called upon in every recursion step.
A crucial observation that we make is that not only does this set contain all $\Lambda$-heavy vertices and no $\Lambda/\polylog{n}$-light vertices, but rather we also know that it does not contain $(2^h\Lambda)$-heavy vertices, because those are handled during previous steps of the recursion. This means that 
we invoke this procedure for $a,b$ that differ only by $\polylog{n}$ and $\eps^2$ factors, and thus we effectively get a running time of $\TO{n^2/\eps^3}$ for counting $h$-cycles through $\Lambda$-heavy vertices.

\paragraph{Roadmap:}
\Cref{sec:template} contains our template for the recursion, and proves its correctness for any graph $\FH$ given implementations of two black boxes, one that finds heavy vertices and another that counts the copies of $\FH$ that contain heavy vertices. \Cref{sec:cycles} proves the running time that our template obtains for $h$-cycles, given the running times of implementations of the two black boxes. We implement our black boxes for $h$-cycles in \Cref{sec:count-heavy,sec:Find Heavy}.

\subsection{Preliminaries}
Let $G$ be a graph on $n$ vertices.
Let $\FH$ be a fixed graph with $h=\BO{1}$ vertices.
For a subgraph $G'\subseteq G$, and a subset of vertices $S$,  denote by $t_{G'}(S)$
number of copies of $\FH$ in $G'$ which intersect $S$.
Denote by $\tau=\tau_{G'}$ the maximal number of copies of $\FH$ in $G'$ in which a vertex participates.
We say that a vertex $v$ is $\Lambda$-heavy (in $G$) if $t_G(v)\geq \Lambda$, and otherwise it is $\Lambda$-light.
We say that a copy $C$ of $\FH$ is $\Lambda$-heavy if $C$ contains at least one $\Lambda$-heavy vertex.
Let $G$ be a graph and $p$ some parameter that could depend on $G$. We denote by $G[p]$ a random induced subgraph of $G$ obtained by keeping each vertex from $G$ independently with probability $p$.
We use $t\apm$ to denote the closed interval $[t(1-\eps),t(1+\eps)]$. We say that a value $\htt=t\apm$ if $\htt\in [t(1-\eps),t(1+\eps)]$.
We assume that $\eps\in(0,1/2]$, which might depend on $n$. If 
$\eps$ is bigger, our algorithm assumes that $\eps\leq 1/2$.
Finally, all logarithms in this paper are base $2$.

\begin{definition}[Fast Matrix Multiplication Definitions]
    We denote the time it takes to compute the product of two matrices of dimension $n^a\times n^b$ and $n^b \times n^c$ by either
    $\MM{n^a,n^b,n^c}$ or $n^{\omega(a,b,c)}$.
    We also abuse the notation and write $\omega=\omega(1,1,1)$,
    and $\omega(k)=\omega(1,k,1)$.
    Note that for any permutation $\pi:[3]\to[3]$ we have
        $\omega(x_1,x_2,x_3) = \omega(x_{\pi(1)},x_{\pi(2)},x_{\pi(3)})\;.$
    In addition to $\omega$, we will also use $\alpha$ to be the largest real number such that $n$ by $n^\alpha$ by $n$ matrix multiplication can be done in $n^{2+o(1)}$ time.
\end{definition}

We will need the following concentration bounds.
\begin{theorem}[Chernoff Bound {\cite{dubhashi2009concentration}}]\label{thm:chernoff 6}
    Let $X_1, \ldots, X_n$ be independent random variables with values in $[0,1]$ and $X=\sum_i X_i$. For $t\geq 6\Exp{X}$, and $\eps>0$ we have
    \begin{align*}
        \Pr{X\geq t}\leq 2^{-t}\;, &&   
        \Pr{X\leq (1-\eps)\Exp{X}}\leq \exp(-\eps^2\cdot \Exp{X}/2)
    \end{align*}
\end{theorem}
\begin{theorem}[Chebyshev's Inequality]\label{ineq:chev}
    Let $X$ be a random variable. Then,
    $\Pr{\abs{X-\Exp{X}}\geq a}\leq \frac{\Var{X}}{a^2}\;.$
\end{theorem}

\section{The Recursive Template}\label{sec:template}
\paragraph{Organization.}
In this section, we present an algorithm for approximating the number of copies of a graph $\FH$ in a graph $G$, denoted by $t$, which builds upon two black boxes.
The first black box, called $\FindHalg$, takes a graph $H$ and a parameter $\Lambda$ as input and computes a superset of the $\Lambda$-\emph{heavy} vertices, excluding any $\Lambda/\polylog{n}$-light vertices. We denote the computed superset of heaviest vertices as $S$.
The second black box, called $\Cialg$, is used to compute an approximation for the number of \emph{heavy-copies} of $\FH$ in $G$, which is the set of all copies that contain at least one vertex from $S$.
In sections \Cref{sec:color code,sec:Find Heavy}, we provide details on how to implement these black boxes as randomized algorithms that produce valid outputs with high probability.

Our algorithm for subgraph approximate counting that uses the specified black boxes consists of two parts: a doubling algorithm called $\algWrapNP$, and a recursive algorithm called $\algCoreNP$, which is the main focus of this section.

\paragraph{The $\algCoreNP$ Algorithm.} The $\algCore{G,\Lambda}$ algorithm takes two parameters: a graph $G$ and a heaviness threshold $\Lambda$.
The output of the algorithm is a value $\htt$, which, with a probability of at least $2/3$, is within the range $t \pm \brak{t\cdot \eps' + \Lambda \cdot \polylog{n}/\eps'} $, where $\eps'$ is the fixed precision parameter of the algorithm.

We next explain how the recursive $\algCoreNP$ algorithm works. The algorithm does the following.
\begin{enumerate}
    \item Find the heaviest vertices using the $\FindHalg$ black box, and denote this set by $V_\Lambda$.
    \item Compute an approximation to the number of heavy copies of $\FH$, i.e., copies of $\FH$ with at least one vertex from $V_\Lambda$, using the $\Cialg$ black box, and denote the output by $\htt_\Lambda$.
    \item Let $H=G[V(G)-V_\Lambda]$, and let $F\gets H[p]$. That is, $F$ is an
          induced subgraph of $H$, where each vertex from $H$ joins $F$ independently with probability $p$.
    \item Make a recursive call to $\algCore{F,\Lambda\cdot p^{\abs{\FH}}}$ and let $\htt_F$ denote its output.
    \item Return $\htt_\Lambda +\htt_F/p^{\abs{\FH}}$.
\end{enumerate}
The analysis of the probability that the algorithm will produce a good approximation appears in the proof of \Cref{lemma2:ih}.

The running time of the algorithm depends on the implementation of the black boxes. In the next section, we analyze the running time of the algorithm.

\paragraph{The $\algWrapNP$ Algorithm.} The $\algWrapNP$ algorithm is a doubling algorithm, which starts with an initial guess for $t$, denoted by $\Thint_0=n^{\abs{\FH}}$. This is the maximal number of copies of $\FH$ an $n$ vertex graph can contain ($h!\cdot\ch{n}{\abs{\FH}}\leq n^{\abs{\FH}}$).
The algorithm then makes $\TO{1}$ calls to $\algCore{G,\Lambda_0}$, where $\Lambda_0\gets \Thint\cdot {\eps}^2/8Q$, where $Q=\Qb$, and computes their median, which we denote by $\htt_0$. 
If $\htt_0\geq \Thint_0$ the doubling algorithm stops and outputs $\htt_0$ as its approximation for $t$.
Otherwise, the guess for the value of $t$ is decreased by a factor of $2$.
The main point here is that the smaller the value of $\Lambda$ given to the recursive algorithm is, the better approximation we get, while simultaneously increasing the running time. The guess $\Thint$ is a guess for the highest heaviness threshold the algorithm can start with to output a good approximation, and not an actual guess for the value of $t$ (although both quantities are related).

Formally, the black boxes that we assume are the following.
\begin{restatable}[$\FindHalg{(G,\Lambda)}$]{blackbox}{reFH}\label{algo:find heavy bb}
    \begin{description}
        \item[Input:] A graph $G$, some parameter $\Lambda$.
        \item[Output:] A subset $V_\Lambda$ of vertices, such that with probability at least $1-\frac{1}{n^4}$, $\forall v\in V(G)$:
              \begin{enumerate}
                  \item If $t_G(v)\geq \Lambda$, then $v\in V_\Lambda$.
                  \item If $t_G(v)\leq \Lambda/\clog$, then $v\notin V_\Lambda$.
              \end{enumerate}
    \end{description}
\end{restatable}

\begin{blackbox}[${\Cialg_{\eps'}(G,V_{\Lambda},a,b)}$]\label{algo:app intersect fh}
    \begin{description}
        \item[Input:] A graph $G$, a precision parameter $\eps'\in[0,\tfrac{1}{2}]$, a subset of vertices $V_\Lambda$, and two real numbers $0<a\leq b$, such that
              $\forall v\in V_\Lambda$ we have $t_G(v)\in[a,b]$.
        \item[Output:]
              $\hat{t}_{\Lambda}$ which satisfies $\Pr{\hat{t}_{\Lambda}=t_G(V_\Lambda)\apmp}\geq 1-\frac{1}{n^4}$.
    \end{description}
\end{blackbox}
It should be noted that this black box cannot be applied to the entire graph, as it might contain a vertex $v$ with $t_G(v) = 0$. Moreover, even if all vertices have $t_G(v) > 0$, employing this black box on the entire graph might result in slower running time. Indeed, in our implementation of the black box, the runtime is contingent on the ratio $b/a$. Therefore, we only employ this black box with values of $a$ and $b$ where $b/a = \mathcal{O}(1/\eps^2)$.
\begin{algorithm}[H]
    \caption{$\algCore{G,\Lambda}$}\label[algorithm]{alg:template rec}
    \setcounter{AlgoLine}{0}
    \KwIn{A graph $G=(V,E)$, a heaviness threshold $\Lambda$, and a precision parameter $\eps'$.}

    \medskip
    $V_\Lambda\gets \FindHalg(G,\Lambda)$ \Comment*{\parbox[t]{3.5in}{$V_\Lambda$ is a superset of $\Lambda$-heavy vertices in $G$ with no $\Lambda/\clog$-light vertices}}
    $a_\Lambda\gets \frac{\Lambda}{\clog}\;,b_\Lambda\gets \Lambda\cdot \frac{8Q}{\eps^2}$ \label{line:l4} \Comment*{\parbox[t]{3.5in}{$Q=\Qb$}}
    \medskip
    $\haT_\Lambda\gets\Cialg_{\eps'}(G,V_\Lambda,a_\Lambda,b_\Lambda)$
    \Comment*{\parbox[t]{3.5in}{$\haT_\Lambda$ is a $\apmp$ approximation for $t_G(V_\Lambda)$ (the number of copies of $\FH$ intersecting $V_\Lambda$)}}
    \lIf{$\Lambda\leq 1$}{\Return{$\haT_\Lambda$}}

    $H\gets G[V-V_\Lambda]$\\
    $p\gets 1/2$\Comment*{\parbox[t]{3.5in}{We keep $p$ instead of $1/2$ for readability}}
    $F\gets H[p]$ \Comment*{\parbox[t]{3.5in}{$F$ is a random subgraph of $H$}}
    $\haT_F\gets \algCore{F ,\Lambda\cdot p^h}$\\
    \smallskip
    \Return{$\haT_{\Lambda}+\haT_F/p^{h}$}
\end{algorithm}

The guarantees for the template are given in the following lemma.

\begin{restatable}[Guarantees for the $\algCoreNP$ Algorithm]{lemma}{LemCoreSimp}\label{lemma:simp}
    For every $\eps\in(0,1/2]$ and every $\Lambda\geq \tau_G\cdot \frac{\eps^2}{8Q}$, we have
    \begin{align*}
         & \Pr{\algCore{G,\Lambda}=t_G(1\pm \eps/2) \pm \Lambda\cdot \tfrac{\log^4(n)}{2\eps}}\geq 2/3\;,
    \end{align*}
    where $\eps'=\neweps$.
\end{restatable}
\begin{remark}
    The depth of the recursion in $\algCore{G,\Lambda}$ is at most $\log_{1/p^h}(\Lambda) + 1$. Since we will only call this algorithm with $\Lambda\leq \abs{V(G)}^h$, we can conclude that the depth of the recursion is at most $\log n + 1$.
\end{remark}
The final algorithm is the wrapper.

\begin{algorithm}[H]
    \caption{$\algWrap{G,\eps}$}\label[algorithm]{alg:template wrap}
    \setcounter{AlgoLine}{0}
    \KwIn{A graph $G=(V,E)$ with $t_G$ copies of $\FH$ and a precision parameter $\eps\leq 1/2$.}
    \KwOut{$\htt$, which is a $\apm$ approximation for $t$ \whp.}
    \medskip
    $\eps'\gets\neweps\;,\;\Thint\gets n^h\;,\;Q\gets \Qb\;,\; \Lambda\gets \Thint \cdot{\eps}^2/Q$\\ %

    \medskip
    \For{$i=0$ to $i=h\log n$}
    {
        $\htt_i\gets \mathsf{Median}\sbrak{\algCore{G,\Lambda/2^i},\rr}\label[line]{line:5med} $\Comment{$\htt_i$ is the median of $\rr$ \Indp\Indp \Indp  independent \indent\hspace{20pt}\indent \hspace{195pt}executions of $\algCore{G,\Lambda/2^i}$.}
        \If{$\haT_i\geq \Thint/2^i$}{
            \Return $\htt_i$
        }
    }
    \Return{Exact deterministic count of $t_G$.}
\end{algorithm}

\renewcommand{\Cialg}{{\mathsf{Count} \mhyphen\mathsf{Heavy}_{\eps'}}}
\begin{restatable}[Guarantees for the $\algWrapNP$ Algorithm]{lemma}{LemWrap}\label{lemma:correctness wrap}
    Let $G$ be a graph with $n$ vertices and $t_G$ copies of $\FH$.
    Fix some $\eps>0$ that may depend on $n$.
    Let $\htt$ denote the output of $\algWrap{G,\eps}$ (specified in \Cref{alg:template wrap}). Then,
    $\begin{aligned}
            \Pr{\htt=t\apm}\geq 1- 1/n^2\;.
        \end{aligned}$
\end{restatable}

The next claim shows that \whp, all calls to the $\FindHalg$ black box generated by the doubling algorithm produce a valid output.
\begin{claim}[The event $\Ebb$]\label{claim:event ebb}
    Run the doubling algorithm.
    Let $\Ebb$ denote the event that all calls made to the $\FindHalg$ black box produce a valid output. Then $\Pr{\Ebb}\geq 1-\frac{\log^4}{n^4}$.
\end{claim}
\begin{proof}[Proof of \Cref{claim:event ebb}]
    The doubling algorithm makes at most $\rr\cdot h\log n$ calls to the recursive algorithm. Each call to the recursive algorithm generates at most $\log n$ calls to the $\FindHalg$ black box.
    Each call produces a correct output w.p. at least $1-\frac{1}{n^4}$. Therefore, by a union bound, the probability that all calls produce a valid output is at least $1-\frac{400h^2\log^3 n}{n^4}\geq 1-\frac{\log^4 n}{n^4}$.
\end{proof}
The next claim shows that if the conditions of the recursive algorithm are met, then so are the conditions of the $\Cialg$ black box.
We need it to claim that all calls to the $\Cialg$ black succeed \whp.

\begin{claim}\label{claim:event count heavy single}
    \sloppy{
    Let $\algCore{G,\Lambda}$ be a call to the recursive algorithm with $\Lambda\geq \tau_G\cdot \frac{\eps^2}{8Q}$, which calls $\Cialg(G,V_\Lambda,a_{\Lambda},b_{\Lambda})$.
    Define an event $\FF_G\triangleq \set{\bigwedge_{v\in V_{\Lambda}}t_{G}(v)\in [a_{\Lambda},b_{\Lambda}]}$. Then, $\Ebb\subseteq \FF_G$.
    }
\end{claim}
\begin{proof}
    We first note that assuming $\Ebb$ occurs, we have that $\set{\bigwedge_{v\in V_{\Lambda}}t_{G}(v)\geq a_{\Lambda}}$ occurs. This follows by the guarantee of the $\FindHalg$ black box which ensures that $V_{\Lambda_j}$ does not contain $\Lambda_j/\clog$-light vertices, where $a_{\Lambda_j}=\Lambda_j/\clog$.
    We next prove that $\Ebb\subseteq \set{\bigwedge_{v\in V_{\Lambda}}t_{G}(v)\leq b_{\Lambda}}$.
    We have that $b_{\Lambda}=\Lambda \cdot \frac{8Q}{\eps^2}\geq \tau_G$, where the inequality is by the assumption on $\Lambda$, and therefore $b_{\Lambda}\geq \tau_G$ (always).
\end{proof}
In \Cref{lemma2:ih,lemma:correctness wrap}, the condition $\Lambda\geq \BoundLam{G}$ is employed solely due to the $\Cialg$ black box; part of its input comprises upper and lower bounds on $\min_{v\in V_{\Lambda'}} t_{G'}(v)$, and $ \;\max_{v\in V_{\Lambda'}} t_{G'}(v)$. Determining the lower bound based on the current threshold value of $\Lambda'$ is straightforward: we set $a_{\Lambda'}\gets \Lambda'/\clog$, which is a valid lower bound based on the guarantees of the $\FindHalg$ black box.
Next, we explain how to find the upper bound.
We differentiate between the initial call to $\algCoreNP$ and the recursive calls. For the recursive calls, determining such an upper bound based on the current threshold value of $\Lambda'$ is straightforward: we set $b_{\Lambda'}\gets \Lambda'/p^h$. This choice is a valid upper bound, supported by the guarantees of the $\FindHalg$ black box, which removed all $\Lambda'/p^h$-heavy vertices from the graph in the preceding call.
For the initial call, we can't set the upper bound $b_{\Lambda'}$ based on the value of $\Lambda'$ in general. Therefore, we establish the claim only for $\Lambda'\geq \tau_G\cdot \frac{\eps^2}{8Q}$ and then set $b\gets \Lambda' \cdot \frac{8Q}{\eps^2}$.
\subsection{Proof of \Cref{lemma:correctness wrap} using \Cref{lemma:simp}}
\newcommand*{\yy}{{y^\ast}}

Recall that $\W=n^h$, and let $\W_i=\W/2^i$.
Let $\yy$ denote the largest integer such that $\W/2^{\yy} \geq t_G/8$.
We emphasize that the value of $\yy$ is not known by the algorithm.

We will define an event $\FFall$, and show that $\Pr{\FFall}\geq 1-\frac{1}{n^2}$.
Then, we prove that the event $\set{\htt=t_G\apm}$ contains the event $\FFall$. This means that if $\FFall$ occurs, then so does $\set{\htt=t_G\apm}$.
This will complete the proof of \Cref{lemma:correctness wrap}.

\medskip\noindent
Let $\htt_{i,j}$ denote the output of the $j$-th execution of $\algCore{G,\Lambda/2^i}$ in \Cref{line:5med}, and let $\htt_i$ denote the median of $\set{\htt_{i,j}}_{j\in[\rr]}$.
We say that $\hat{s}\in\set{\htt_{i,j},\htt_{i}}$ is
\emph{good} if $\hat{s}= t_G\pm\frac{\eps}{2}\brak{t_G + \W_i/4}$.
\noindent We denote by $\FFall'$ the event that for $0\leq i\leq \yy$,
the median $\htt_i$ is good.
Let $\FFall\triangleq \FFall'\cap \Ebb$.
\begin{proposition}\label{prop:ffall}
    $\Pr{\FFall}\geq 1-\frac{1}{n^2}$.
\end{proposition}
\begin{proof}
    We will show that $\Pr{\FFall'\mid \Ebb}\geq \fpr{\log^4 n}{4}$, and use the law of total probability to get the desired bound on $\FFall$, as
    $\Pr{\Ebb}\geq \fpr{\log^4 n}{4}$ by \Cref{claim:event ebb}.
    To prove $\Pr{\FFall'\mid \Ebb}\geq \fpr{\log^4 n}{4}$, we fix some $i\leq \yy$.
    If we show that $\htt_i$ is good with probability at least $1-1/n^4$, we can apply a union bound over all $i\leq \yy$ and complete the proof.
    We show that for every $j$ we have that $\htt_{i,j}$ is good with probability at least $2/3$.
    We then get that $\htt_i$ is also good with probability at least $1-1/n^5$. This can be shown by a standard Chernoff's inequality and is sometimes referred to as ``the median trick'' (See \Cref{claim:med trick0} for a proof).

    To show that $\htt_{i,j}$ is good with probability at least $2/3$, we want to apply \Cref{lemma:simp}, but we must ensure that the conditions of the lemma are met, that is, that the each direct call the doubling algorithm makes to the recursive algorithm, denoted by $\algCore{G,\Lambda}$, satisfies $\Lambda\geq \tau_{G}\cdot \frac{\eps^2}{8Q}$.
    We know that
    $\Lambda=\frac{\W}{2^i}\cdot\frac{\eps^2}{Q}$ as we are in the $i$-th iteration of the doubling algorithm. Then,
    $ \Lambda
        = \frac{\W}{2^i} \cdot \frac{\eps^2}{Q}
        \geq \frac{\W}{2^{\yy}}\cdot \frac{\eps^2}{Q}
        \geq t_G\cdot \frac{\eps^2}{8Q}
        \geq \tau_G\cdot \frac{\eps^2}{8Q}
    $.
    The penultimate inequality follows from the definition of $\yy$.
\end{proof}

We are now ready to prove \Cref{lemma:correctness wrap}.
\begin{proof}[Proof of \Cref{lemma:correctness wrap}]
    Let $\htt$ denote the output of the doubling algorithm.\\
    We prove that $\set{\htt=t_G\apm}\supseteq \FFall$. That is, we prove a \emph{deterministic} claim: If for each $0\leq i \leq \yy$ we have that $\htt_i$ is good, then $\htt=t_G\apm$.
    Assume $\FFall$ occurs and fix some index $i\leq \yy$.
    Informally, we first prove that if $\Thint_i$ is too large with respect to $t_G$, namely $\Thint_i \geq 4t_G$, then
    $\htt_i$ is smaller than $\Thint_i$. Therefore, \Cref{alg:template wrap} does not stop; instead, it retries with a finer threshold parameter, $\Thint_{i+1}=\Thint_i/2$. This implies that we need to prove $\htt=t_G\apm$ only when $\Thint_i \leq 4t_G$.
    We then show that if $\Thint_i$ is smaller than $4t_G$ then indeed $\htt_i=t_G\apm$. However, this is not sufficient to ensure that the algorithm stops. Therefore, we further show that if $\Thint_i$ is smaller than $t_G/2$ then $\htt_i\geq \Thint_i$ and therefore the algorithm always stops, and by the prior claim gives a correct estimation $\htt=t_G\apm$.

    \medskip
    The proof shows this by a case analysis on $\Thint_i$.
    We fix some $i\leq \yy$, which means $\W_i\geq t_G/8$.
    Let $L_i\triangleq(t_G(1\pm\eps/2)\pm \Thint_i\cdot \eps/8)=t_G \pm\frac{\eps}{2}\brak{t_G+\W_i/4}$. That is, $L_i$ denotes an interval that contains $\htt_i$, as we assumed $\FFall$ occurs.
    \begin{enumerate}
        \item For $i$ such that $\W_i\geq 4t_G$, we show that $\htt_i\leq \max(L_i)< \W_i$. This is because $\max(L_i)\leq \frac{\W_i}{2} + \frac{\W_i}{8}< \W_i$.
        \item For $i$ such that $\W_i\leq 4t_G$, we show that $\htt_i=t_G\apm$.
              This follows as $\frac{\eps}{2}\brak{t_G+\W_i/4}\leq \eps \cdot t_G$.
        \item For $i$ such that $\W_i\leq t_G/2$, we show that $\htt_i\geq \W_i$. To see this, we need to show that $t_G/2\leq \min(L_i)$, which follows since $\min(L_i)\geq t_G\brak{1-\eps(1/2 + 1/8)}> t_G/2$ as $\eps\leq 1/2$.
    \end{enumerate}
    To summarize, assuming $\FFall$ occurs, the doubling algorithm stops in the $i$-th iteration, for $i\in[\yy -5,\yy]$, with a correct approximation, as $\W/2^{\yy - 5}\leq 4t_G$.
    
\end{proof}
Later, we will need the following conclusion to analyze the running time.
\begin{conclusion}\label{conclusion:min thint}
    If the doubling algorithm calls $\algCore{G',\Lambda'}$ then
    $\FFall\subseteq \set{\Lambda'\geq t_G \cdot \frac{\eps^2}{8Q}}$.
\end{conclusion}
\begin{proof}
    Assuming $\FFall$ occurs, the doubling algorithm stops in the $i$-th iteration, for $i\in[\yy -5,\yy]$, and therefore
    $\Lambda'=\frac{\W}{2^i} \cdot \frac{\eps^2}{Q} \geq t_G \cdot \frac{\eps^2}{8Q}$.
\end{proof}

\subsection{%
    Proving \Cref{lemma:simp} Using a More Refined Version}\label{ssec:params}
\begin{restatable}[{$\DD$}]{definition}{defDepth}\label{def:Depth}
    We define the depth of the call $\algCore{F,\Lambda}$ as $\DD[\Lambda]=\max\brak{\ceil*{\log_{1/p^h}(\Lambda)},0}$. 
\end{restatable}
We assume $\Lambda\leq (h!)\cdot {\ch{n}{h}}\leq n^h$, and therefore that $\DD[\Lambda]\leq \log n +1$.
To prove \Cref{lemma:simp}, we state and prove a more refined version of the guarantees of the template algorithm, as follows.
\begin{restatable}[Induction Hypothesis]{lemma}{LemInduct}\label{lemma2:ih}
    Given a graph $G$ and $\eps$, we set $p=1/2$, and $\eps'=\neweps$.
    Then, for any $\Lambda\geq \BoundLam{G}$, and any $\K=o(n)$ that could depend on $n,\DD$ and $p$, we have
    \begin{align*}
        \Pr{\algCore{G,\Lambda}=t_G(1\pm \eps')^{\DD} \pm 2\DD\K \cdot \Lambda/\eps' }
        \geq 1- \frac{4h\DD}{\K p^h}\;.
    \end{align*}
\end{restatable}

We use \Cref{lemma2:ih} to prove \Cref{lemma:simp}, which we restate here for convenience.
\LemCoreSimp*

We use $p\triangleq1/2$ throughout the section. This implies that $\DD[\Lambda]\leq \log n$.
We also set $\K= \frac{\log^2(n)}{16}$.
We choose finer error parameter $\eps'$ as a function of the desired error $\eps$, as follows $\eps'\gets \neweps$. This will ensure that
$\begin{aligned}
        (1\pm \eps')^{\log n}\subseteq(1\pm \eps/2)\end{aligned}$ and that
$\begin{aligned}
        (1+ \eps')^{\log n}\leq 2
    \end{aligned}$.
This refined choice of $\eps'$ will affect only lower order terms of the running time.
\begin{proof}[Proof of \Cref{lemma:simp} using \Cref{lemma2:ih}]
    First, for the multiplicative overhead in the approximation, we prove that $(1\pm \eps')^{\log n}\subseteq(1\pm \eps/2)$ and that $(1+ \eps')^{\log n}\leq 2$, as follows.
    \paragraph{Bounding $(1+\eps')^{\log n}$.}
    It holds that
    $ (1+\eps')^{\log n} =\brak{1+\frac{\eps}{4\log n}}^{\log n}\leq \exp\brak{\frac{\eps}{4}}$, which is at most $1+\frac{\eps}{4} +\brak{\frac{\eps}{4}}^2$ since for any $y\leq 1$ we have $\exp(y)\leq 1 + y + y^2$ \cite[Lemma 1.4.2 (b)]{doerr2019theory}. This is at most $1+ \frac{\eps}{2}$ since $\eps\leq 1$.

    Since $\eps\leq 1$, the above also implies that  $(1+\eps')^{\log n}\leq 2$.
    \paragraph{Bounding $(1-\eps')^{\log n}$.} By Bernoulli's inequality for $-\eps$ we have $(1-\eps')^{\log n}  \geq 1-\eps'\log n=1-\eps/4$.
    \medskip

    Next, for the additive overhead, note that by setting $\K=\log^2 n/16$,
    for any $\Lambda$ and $\eps'$, it holds that
    $2\DD\K\cdot\Lambda/\eps'= \frac{\log^3 n}{8\eps'}\cdot \Lambda=
        \frac{\log^4 n}{2\eps}\cdot \Lambda   \;.$
    \medskip\noindent
    For the error probability, note that for $\K=\log^2 n/16$, we have:
    $1-\frac{4h\DD}{\K p^h}=1-\frac{64h\log n}{2^{-h}\cdot \log^2 n }=1-\frac{\BO{1}}{\log n}\geq 2/3$,
    which completes the proof.
\end{proof}

\subsection{Proof of \Cref{lemma2:ih}}
It remains to prove \Cref{lemma2:ih}.
\LemInduct*

\bigskip\noindent We prove \Cref{lemma2:ih} using induction on the depth of the recursion.

\paragraph*{Base Case:}
Let $\Em{2}$ denote the event that
$\set{\algCore{G,\Lambda}=t_G(1\pm \eps')^{\DD} \pm 2\DD\K \cdot \Lambda/\eps' }$.
We show that $\Pr{\Em{2}}\geq 1- \frac{4h\DD}{\K p^h}$, when $\DD=0$, i.e., $\Lambda\leq 1$.
We show something stronger -- that $\Pr{\htt=t_G\apmp}\geq \pnp$.
Note that when $0<\Lambda\leq 1$, there is only one superset of $\Lambda$-heavy vertices, with no $\Lambda/\clog$-light vertices -- the set of all vertices $v$ with $t_G(v)>0$. 
The $\FindHalg$ black box will produce this set with probability at least $\fpr{1}{4}$. 
Assuming $V_\Lambda$ was computed correctly, we have $t_G=t_G(V_\Lambda)$, and therefore $\Pr{\htt_\Lambda=\apmp t}\geq \fpr{1}{4}$ by the $\Cialg$ black box. 
Overall, we get that $\Pr{\htt_\Lambda=\apmp t}\geq \fpr{2}{4}$, which completes the base case.

\paragraph*{Induction Hypothesis:}
If $\Lambda\geq \tau_G\cdot \frac{\eps^2}{8Q}$, then
$$\begin{aligned}
        \quad \Pr{\algCore{G,\Lambda}=t_G(1\pm \eps')^{\DD} \pm 2\DD\K \cdot \Lambda/\eps' }
        \geq 1- \frac{4h\DD}{\K p^h}\;.
    \end{aligned}$$

\paragraph*{Step:}
We first define some notation and events.
We have three graphs $F\subseteq H\subseteq G$, where $G$ is the input graph, $H$ is an induced subgraph of $G$ without the $\Lambda$-heavy vertices, and $F$ is a random induced subgraph of $H$ obtained by keeping each vertex of $H$ independently with probability $p$.
Let $\haT_G$ denote the output of $\algCore{G,\Lambda}$, and let $\haT_F$ denote the output of $\algCore{F,\Lambda \cdot p^h}$.
We use $t_\Lambda$ to denote $t_G(V_\Lambda)$, and let $\htt_\Lambda$
denote the output of $\Cialg(G,V_\Lambda)$.
Note that $\DD\geq \DDF+1$ (unless $\DD=1$, in which case $\DD=\DDF=1$).
\paragraph{Intuition.}
We compute two values, $\htt_\Lambda$ and $\htt_F$.
We then output $\htt=\htt_\Lambda + \htt_F/p^h$ as an approximation for $t_G$.
By the black box guarantees, we have that $\htt_\Lambda$ is a good approximation for $t_\Lambda$. What is left is to show that $\htt_F/p^h$ is a good approximation for $t_H$. We split this into two parts. First, we show that $t_F/p^h$ is ``close'' to the value of $t_H$ (See \Cref{cor2:e3}).
Next, we use the induction hypothesis, to show that $\htt_F$ is a good approximation of $t_F$. We need to show that the ``composition'' of these approximations is also good.

Recall that $\Ebb$ denotes the event that ${V_\Lambda}$ contains a superset of the $\Lambda$-heavy vertices, without any $\Lambda/\clog n$-light vertices.
    {\bfseries{We prove the correctness of the algorithm under the assumption that $\Ebb$ occurs}}.
\begin{enumerate}
    \item $\Em{1}\triangleq\set{\haT_\Lambda=t_\Lambda\apmp}$. The heavy copies of $\FH$ are approximated correctly.
    \item $\Em{2}\triangleq\set{\haT_F=t_F\apmp^{\DDF} \pm 2\DDF\cdot \K \cdot  (\Lambda\cdot p^h)/\eps'}$. This is the event specified in the induction hypothesis.
    \item $\Em{3}\triangleq\set{t_F/p^h=t_{H}\apmp\pm \K \cdot \Lambda/\eps'}$. This is a concentration bound on $t_F$.
    \item $\Em{4}\triangleq\set{\haT_F/p^h=t_H\apmp^{\DD}\pm 2\DD\K\cdot \Lambda/\eps'}$. This event contains $\Em{2}\cap \Em{3}$.
    \item $\Em{5}\triangleq\set{\haT_G= t_G\apmp^{\DD}\pm 2\DD\K\cdot \Lambda/\eps'}$. This is the event specified in \Cref{lemma2:ih}.
\end{enumerate}

\bigskip \noindent
\Cref{lemma2:ih} requires us to show that $\Pr{\Em{5}}\geq 1-\frac{4h\DD}{\K p^h}$. We will show this by proving that $\Em{1}\cap \Em{2}\cap \Em{3}\subseteq \Em{5}$, which implies that it is sufficient to prove that $\Pr{\Ebb\cap \Em{1}\cap \Em{2}\cap \Em{3}} \geq 1-\frac{4h\DD}{\K p^h}$. To show the latter, we will show that $\Pr{\Ebb\cap \Em{1}\cap \Em{2}}\geq 1-\frac{4h\DDF + 2}{\K p^h}$ and that $\Pr{\Em{3}}\geq 1-\frac{h+1}{\K p^h}$. Summing up the two error probabilities in the above expressions will give the desired result.

\begin{claim}\label{claim2:all events}
    $\Em{1}\cap \Em{2}\cap \Em{3}\subseteq \Em{5}$.
\end{claim}
\begin{proof}
    We first show $\Em{1}\cap\Em{4}\subseteq \Em{5}$.
    \begin{align*}
        \haT_G & \overset{def}{=}\quad \haT_\Lambda +\haT_F/p^h                                                                                                             \\
               & \hspace{-4pt}\overset{\Em{1}\cap\Em{4}}{\subseteq}\hspace{-4pt}\quad t_\Lambda(1\pm \eps') +t_H(1 \pm \eps')^{\DD} & \pm \brak{2\DD \K\cdot \Lambda/\eps'} \\
               & \subseteq  (t_\Lambda + t_H) (1\pm \eps')^{\DD}                                                                    & \pm \brak{2\DD \K\cdot \Lambda/\eps'} \\
               & =  t_G (1\pm \eps')^{\DD}                                                                                          & \pm \brak{2\DD \K\cdot \Lambda/\eps'}
    \end{align*}

    We next show that $\Em{2}\cap\Em{3}\subseteq \Em{4}$.
    \begin{align*}
        \dfrac{\haT_F}{p^h}
         & \overset{\Em{2}}{=}\;
        \dfrac{t_F\apmp^{\DDF}\pm \brak{2\DDF\K p^h \Lambda/\eps'}}{p^h}                                                                                          \\
         & = \dfrac{t_F\apmp^{\DDF}}{p^h}                                                                                       & \pm\brak{2\DDF\K \Lambda/\eps'} \\
         & \overset{\Em{3}}{=}\; \sbrak{t_H \apmp\pm \K \cdot \Lambda/\eps'}\apmp^{\DDF}                                        & \pm\brak{2\DDF\K \Lambda/\eps'} \\
         & \overset{(\star)}{\subseteq} t_H\apmp^{\DD} \pm \brak{2\K\Lambda/\eps'}                                              & \pm\brak{2\DDF\K \Lambda/\eps'} \\
         & \subseteq t_H\apmp^{\DD}                                                      \pm\brak{2\DD\cdot\K \Lambda/\eps'}\;.
    \end{align*}
    The $(\star)$ transition is true because we chose sufficiently small $\eps'$ such that $(1\pm\eps')^{\DDF}\subseteq \apm\subseteq [-1/2,3/2]$.
    This completes the proof of \Cref{claim2:all events}.
\end{proof}

As explained, \Cref{claim2:all events} shows that in order to prove \Cref{lemma2:ih} it is sufficient to prove that
\begin{align}
    \Pr{\Ebb\cap \Em{1}\cap \Em{2}\cap \Em{3}} \geq 1-\frac{4h\DD}{\K p^h}\;.\label{eq:event}
\end{align}
We split the proof into two parts.
We first show that $\Pr{\Ebb\cap \Em{1}\cap \Em{2}}\geq 1-\frac{4h\DDF + 2}{\K p^h}$, which is the simpler task, and then we show that $\Pr{\Em{3}}\geq 1-2h/(\K\cdot p^h)$. 
\begin{claim}\label{claim:w3}
    $\Pr{\Em{2}\mid \Ebb}\geq 1-\frac{4h\DDF}{\K p^h}$
\end{claim}
\begin{proof}
    This is the induction hypothesis. But to use it, we need to make sure that $\Lambda\cdot p^h\geq \tau_F\cdot \frac{\eps^2}{8Q}$.
    This follows by the fact that $F\subseteq H\subseteq G$, where $F$ does not have $\Lambda$-heavy vertices (Assuming $\Ebb$ occurs).
    Therefore, $\tau_F\leq \Lambda$ which means $\tau_F \cdot \frac{\eps^2}{8Q} \leq \Lambda\cdot p^h$, and therefore we can use the induction hypothesis.
\end{proof}
\begin{claim}\label{claim:y17}
    $\Pr{\Ebb\cap \Em{1}\cap \Em{2}}\geq 1-\frac{4h\DDF + 2}{\K p^h}$.
\end{claim}
\begin{proof}[Proof of \Cref{claim:y17}]
    It holds that $\Pr{\Ebb}\geq 1-1/n^2$ by the promise of the $\FindHalg$ black box. Further, $\Pr{\Em{1}}\geq 1-1/n^2$ by \Cref{claim:event count heavy single}, and the guarantee of $\Cialg$. By \Cref{claim:w3}, we have that $\Pr{\Em{2}\mid \Ebb}\geq 1-\frac{4h\DDF}{\K p^h}$.
    Therefore, using the fact that $\K =o(n)$ and a union bound, we get that $\Pr{\Ebb\cap \Em{1}\cap \Em{2}}\geq 1-\frac{4h \DDF +2 }{\K p^h}$.
\end{proof}

We are left with proving that $\Pr{\Em{3}}\geq 1-\frac{4h-2}{\K p^h}$, and then by a union bound %
we get \Cref{eq:event} (since $\DDF+1=\DD$).
\begin{claim}\label{cor2:e3}
    $\Pr{\Em{3}}\geq 1-\frac{h+1}{\K p^h}$.
\end{claim}
To prove \Cref{cor2:e3}, we first bound $\Exp{(t_F)^2}$, i.e., the second moment of the number of copies of $\FH$ in the induced random subgraph $F$. We then use Chebyshev's inequality.
Recall that we say that a vertex $v$ is $\Lambda$-heavy over a graph $G$ if it participates in at least $\Lambda$ copies of $\FH$ in $G$, and otherwise, we call it $\Lambda$-light.
We bound $\Exp{(t_F)^2}$ in terms of the number of triangles in $H$, denoted by $t_H$ (where $t_H\leq t_G$), and $\tau_H$, where $\tau_H$ is the smallest integer such that all vertices in $H$ are $(\tau_H+1)$-light.%
For brevity, we use $\tau=\tau_H$.
Note that every vertex in $H$ is $\Lambda$-light, and therefore $\tau\leq \Lambda$.

\begin{definition}
    For $i\in [h+1,2h]$, let $\FH^i$ denote the number of pairs of copies of $\FH$ in $H$ which share exactly $2h-i$ vertices, or equivalently, the union of their vertex sets contains exactly $i$ vertices.
\end{definition}
Note that $\Exp{(t_F)^2} = t_H\cdot p^h +\sum_{i= h+1}^{2h}\FH^i\cdot  p^i$.
The next claim will help us bound the term $\FH^i$.
\begin{claim}\label{claim:H share i vertices}
    For every
    $i\in[h+1,2h-1]$, it holds that $\FH^i\leq h\cdot \tau_H \cdot t_H.$
\end{claim}
\begin{proof}[Proof of \Cref{claim:H share i vertices}]
    Let $(C_1,C_2)$ be two copies of $\FH$ sharing $2h-i$ vertices, where $i\in [h+1,2h-1]$.
    Let $v_0$ be the vertex with minimal index among the vertices that $C_1$ and $C_2$ share.
    We can encode $(C_1,C_2)$ using $(C_1,j,r)$ where $j\in[h]$ and $r\in [\tau_H]$. To see this, choose $j$ such that $v_0$ is the $j$-th vertex in $C_1$, and let $r$ indicate that $C_2$ is the $r$-th copy of $\FH$ in which $v_0$ participates.
    We need only $[\tau_H]$ indices for $r$ as every vertex participates in at most $\tau_H$ copies of $\FH$.
\end{proof}
\begin{proof}[Proof of \Cref{cor2:e3}]
    First note that $\Exp{t_F}=t_H\cdot p^h$, as each copy of $\FH$ in $H$ will also appear in $F$ if all of its vertices where sampled into $F$, which occurs with probability $p^h$.
    We next bound $\Exp{(t_F)^2}$. By \Cref{claim:H share i vertices} we get
    \begin{align*}
        \Exp{(t_F)^2}
         & = t_Hp^h +\sum_{i= h+1}^{2h}\FH^i\cdot  p^i
         &                                                            & \leq  t_H p^{h} + h \tau_H \cdot t_H\sum_{i= h+1}^{2h-1} p^i + (t_Hp^{h})^2 \\
         & \leq t_Hp^{h}\brak{1+h\tau_H\sum_{i=1}^{2h-1}p^i + t_Hp^h}
         &                                                            & \leq t_Hp^{h}\brak{1+2h\tau_H + t_Hp^h}\;.
    \end{align*}
    We also get
    \begin{align*}
        \Var{t_F}=\Exp{(t_F)^2}-\Exp{t_F}^2\leq t_Hp^{h}\brak{1+2h\tau_H + t_Hp^h} - (t_Hp^h)^2\leq t_Hp^{h}\brak{1+2h\tau_H}\leq t_Hp^{h}\brak{1+2h}\tau_H\;.
    \end{align*}
    Next, we use Chebyshev's inequality:
    $
        \Pr{\abs{Y-\Exp{Y}}\geq \psi}\leq \frac{\Var{Y}}{\psi^2}\;.
    $

    Recall that
    $\Em{3}\triangleq\set{t_F/p^h=t_{H}\apmp\pm \K \cdot \Lambda/\eps'}$.
    We can also write $\Em{3}$ as follows $\set{\abs{t_F-\Exp{t_F}}\leq \sigma}$ where $\sigma=p^h(t_H\eps' + \K\Lambda/\eps')$ as:
    \begin{align*}
        \Em{3}
         & =\set{t_F/p^h-t_{H} = \pm \brak{t_H\eps'+\K \cdot \Lambda/\eps'}}                \\
         & =\set{\abs{t_F/p^h-t_{H}} = t_H\eps'+\K \cdot \Lambda/\eps'}                     \\
         & = \set{\abs{t_F-t_{H}\cdot p^h} \leq  p^h\brak{t_H\eps'+\K \cdot \Lambda/\eps'}} \\
         & = \set{\abs{t_F-\Exp{t_F}} \leq  \sigma}
    \end{align*}
    We apply Chebyshev's inequality on $t_F$ and get that
    \begin{align*}
        \Pr{\Em{3}}
        =\Pr{\abs{t_F-\Exp{t_F}}\leq \sigma}
        =1-\Pr{\abs{t_F-\Exp{t_F}}> \sigma},
    \end{align*}
    where
    \begin{align*}
        \Pr{\abs{t_F-\Exp{t_F}}\geq \sigma}
        \leq \frac{\Var{t_F}}{\sigma^2}
        \leq \frac{t_Hp^{h}\brak{1+2h}\tau_H}{2p^{2h}\cdot t_H\K\Lambda}
        \leq \frac{1+2h}{2p^{h}\cdot \K}\leq \frac{h+1}{p^{h}\cdot \K}\;.
    \end{align*}
    The penultimate inequality uses the fact that $\Lambda\geq \tau_H$, which is true because $H$ does not have $\Lambda$-heavy vertices.
    Overall, we get  $\Pr{\Em{3}}\geq 1-\frac{h+1}{p^{h}\K}$
    which completes the proof of \Cref{cor2:e3}.
\end{proof}
This completes the proof of \Cref{lemma2:ih}, and therefore also \Cref{lemma:simp,lemma:correctness wrap}.

\section{Application: Approximating the Number of $h$-Cycles}
\label{sec:cycles}
\newcommand{\f}{\Psi_{\algCoreNP}}
\newcommand{\gh}{{1/(h-2)}}

In this section, we analyze the running time of the recursive and doubling algorithm, when the counted subgraph is an $h$-cycle for $h=\BO{1}$.
For this task, we assume the black boxes can be implemented in specific runtime, as stated in the following lemma which is proven in the next sections.
\begin{lemma}\label{lemma:assumption}
    Let $G$ be a graph with $n$ vertices, let $\FH$ be an $h$-cycle for some $h=\BO{1}$, and let $\eps'$ be some parameter. Then, 
    \begin{itemize}
        \item Each call to $\FindHalg(G,\Lambda)$ can be implemented in time $\TO{\MM{n,n,\frac{n}{\Lambda^{1/(h-2)}}}}$.
        \item Each call to $\Cialg(G,V_{\Lambda},a_{\Lambda},b_{\Lambda})$ can be implemented in time $\TO{\MM{n,n,1}\cdot \frac{b_{\Lambda}}{a_{\Lambda}\cdot \eps}}=\TO{n^2\cdot \frac{b_{\Lambda}}{a_{\Lambda}\cdot \eps}}$.
    \end{itemize}
\end{lemma}
We will prove the following theorem.

\ThmZero* 

We denote the running time of the recursive algorithm $\algCore{G,\Lambda}$ by $\f(G,\Lambda,\eps')$, and prove that given \Cref{lemma:assumption}, it is bounded by $\TO{\MM{n,n,\frac{n}{\Lambda^{1/(h-2)}}} + \frac{n^2}{\eps^3}}$ with probability at least $1-\exp(-\log^2 n)$ (in \Cref{prop:rt bound} below).
Then, we show that the complexity of the doubling algorithm is bounded by
$\TO{\frac{1}{\eps^3}\cdot \MM{n,n,\frac{n}{t^{1/(h-2)}}}}$ with probability at least $\pnp$ (in \Cref{thm:main0}).
\begin{proposition}\label{prop:rt bound}
\sloppy{
    Let $G$ be a graph with at most $n$ vertices. Then,
    $\f\brak{G,\Lambda,\eps'}=\TO{\MM{n,n,\frac{n}{\Lambda^{1/(h-2)}}} + \frac{n^2}{\eps^3}}$, with probability at least $1-\exp(-\log^2 n)$.
}
\end{proposition}

\begin{proof}[Proof of \Cref{thm:main0} using \Cref{prop:rt bound}]
    Consider the $\algWrap{G,\eps}$. Its correctness follows from \Cref{lemma:correctness wrap}. We analyze its running time given \Cref{lemma:assumption}.
    Let $\Psi_{\algWrapNP}(G,\eps)$ denote the running time of $\algWrap{G,\eps}$.
    The doubling algorithm makes a sequence of $y=\TO{1}$ calls to the recursive algorithm, with the parameters $\set{(G,\Lambda/2^i, \eps')}_{i=0}^y$.
    By \Cref{conclusion:min thint} we have that for every $i$, $\Lambda/2^i\geq t_G\cdot\frac{ \eps^2}{8Q}$ with probability at least $1-\frac{h\log n}{n^4}$. This allows us to bound the running time of the doubling algorithm as follows:
    \begin{align*}
        \Psi_{\algWrapNP}(G,\eps)=
          & \sum_{i=0}^{y}\f\brak{G,\frac{\Lambda}{2^i},\eps'}                                                      \\
        = & \sum_{i=0}^{y}\TO{\MM{n,n,\frac{n}{(\Lambda/2^i)^{1/(h-2)}}} + \frac{n^2}{\eps^3}}                      \\
        = & \sum_{i=0}^{y}\TO{\MM{n,n,\frac{n}{\brak{\frac{t\cdot {\eps}^2}{8Q}}^{1/(h-2)}}}  + \frac{n^2}{\eps^3}} \\
          & = \TO{\frac{1}{\eps^3}\cdot \MM{n,n,\frac{n}{t^{1/(h-2)}}}}
    \end{align*}
    The first equality follows from the definition of $\Psi_{\algWrapNP}(G,\eps)$.
    The second equality follows by \Cref{prop:rt bound}.
    The third equality follows because $\MM{n,n,x}\leq \MM{n,n,y}$ for $x\leq y$, and the fact that $\Lambda/2^y\geq t_G\cdot\frac{ \eps^2}{8Q}$ by \Cref{conclusion:min thint}.
    The last equality follows because for every $x\geq 1$, we have $\MM{n,n,x}\geq n^{2}$, as only reading the input takes $\BO{n^2}$ time. 
    Therefore, $\MM{n,n,x} + n^2/\eps^3\leq \MM{n,n,x}/\eps^3$.

    Therefore, the running time of the doubling algorithm is $\TO{\frac{1}{\eps^3}\cdot \MM{n,n,\frac{n}{t^{1/(h-2)}}}}$ with probability at least $1-\brak{\frac{h\log n}{n^4} + \frac{\TO{1}}{n^{\log n}}}\geq \pnp$, which completes the proof.
\end{proof}

For the rest of this section, we prove \Cref{prop:rt bound}.
We unroll the recursion of the algorithm:
it makes a single call to $\FindHalg$, then a single call to $\Cialg$, and then a recursive call.
The algorithm's runtime can be analyzed by bounding $\set{\FindHalg(G_k,\Lambda_k)}_{k=0}^r$, and $\set{\Cialg(G_k,V_{\Lambda_k},a_{\Lambda_k},b_{\Lambda_k})}_{k=0}^r$,
where $r$ is the recursion depth and $G_k$ denotes the input graph for the $k$-th call of the recursive algorithm, where $G_0=G$.
Let $\Lambda_k=\Lambda_0\cdot p^{kh}$ denote the heaviness threshold, and let $V_{\Lambda_k}$ denote the set of $\Lambda_k$-heavy vertices in the $k$-th iteration.
We bound the running time of the $k$-th call to each of the black boxes.

The following claim address the total running time of all calls to the $\Cialg$ black-box.
\begin{claim}\label{conclusion:z1}
    The total running time for all calls to  the $\Cialg$ black-box is  $\TO{n^2/\eps^3}$.
\end{claim}
\begin{proof}
    Fix some $k$.
    By \Cref{lemma:assumption}, the call $\Cialg(G_k,V_{\Lambda_k},a_{\Lambda_k},b_{\Lambda_k})$ takes $\TO{n^2\cdot \frac{b_{\Lambda_k}}{a_{\Lambda_k} \cdot \eps}}$ time, as $G_k$ has at most $n$ vertices, and $\frac{b_{\Lambda_k}}{a_{\Lambda_k}}=\TO{1/{\eps}^2}$, as 
    $a_\Lambda= \frac{\Lambda}{\clog}\;,b_\Lambda= \Lambda\cdot \frac{8Q}{\eps^2}$, where $Q=\Qb$ (See \Cref{line:l4}).
We conclude that the total running time of all of the calls
$\set{\Cialg(G_k,V_{\Lambda_k},a_{\Lambda_k},b_{\Lambda_k})}_{k=0}^r$ takes at most $r$ times the runtime of a single call. As $r\leq \log n$, we get that all calls take a total of $\TO{n^2/\eps^3}$ time.
\end{proof}

\medskip\noindent Next, we bound the running time of the $k$-th call to the $\FindHalg$ black-box.
Note that for any $k$, the call $\FindHalg(G_k,\Lambda_k)$ takes $\TO{\MM{\abs{V(G_k)},\abs{V(G_k)},\abs{V(G_k)}/\Lambda_k^\gh}}$ by \Cref{lemma:assumption}.

\smallskip\noindent We first provide a concentration bound on the number of vertices in $G_k$.
\begin{claim}\label{claim:Gk vertices}
    For any $k$, the following hold: 
    \begin{enumerate}
        \item If $np^k\geq \log^3 n$ then $\Pr{\abs{V(G_k)}\leq 2np^k}\geq 1-\exp(-\log^2 n)$.
        \item If $np^k\leq \log^3 n$ then $\Pr{\abs{V(G_k)}\leq \log^8 n}\geq 1-\exp(-\log^2 n)$.
    \end{enumerate}
\end{claim}
\begin{proof}
\renewcommand{\whp}{with probability at least $1-\exp(-\log^2 n)$, }
    Fix $k$.
    Note that $G_k$ is a random variable, and so is its number of vertices $X_k$. We have $\Exp{X_k}=np^k$.
    We use a Chernoff bound (see \Cref{thm:chernoff 6}).
    \begin{enumerate}
        \item If $np^k\geq\log^3 n$, then $\Pr{X_k> 2\Exp{X_k}}\leq 2^{-2\Exp{X_k}/6}=2^{-np^k/3}\leq \exp\brak{-\log^2 n}$.
        \item If $np^k\leq \log^3 n$, then 
        $\Pr{X_k\geq\log^8 n}\leq \Pr{X_k\geq \Exp{X_k}\cdot \log^5 n}\leq 2^{-\log^5 n/6}\leq \exp(-\log^2 n)$.
    \end{enumerate}
\end{proof}
The next claim is on matrix multiplication. It will allow us to show that the bound on the running time of the first call to the $\FindHalg$ black box
also applies to subsequent calls.
\begin{claim}\label{claim:z mm}
    For any $k\geq 0$, we have
    $
        \MM{n p^k,n p^k,n p^k/\brak{\Lambda\cdot p^{hk}}^\gh}       \leq \MM{n,n,{n/\Lambda^\gh}}
    $.
\end{claim}
\begin{proof}
    \begin{align*}
         & \MM{n p^k,n p^k,n p^k/\brak{\Lambda\cdot p^{hk}}^\gh}            \\
         \leq\; &\MM{n,n,\frac{n}{\Lambda^\gh} \cdot p^{k\cdot (3-h/(h-2))}} \\
         \leq\; &\MM{n,n,{n/\Lambda^\gh}}
    \end{align*}
    The first inequality follows by \Cref{claim:a2}.
    Note that $3-h/(h-2)= 2\frac{h-3}{h-2}$ is non-negative for $h\geq 3$
    Therefore, $p^{k\cdot (3-h/(h-2))}\leq 1$ and the last inequality follows.
\end{proof}

Next, we bound the running time of the $k$-th call to the $\FindHalg$ black box.
\begin{claim}\label{claim:f bound}
    Fix some $k$. The call
    $\FindHalg(G_k,\Lambda_k)$ takes at most $\BO{\MM{n,n,n/\Lambda^\gh}}$ time, with probability at least $1-\exp(-\log^2 n)$.
\end{claim}
\begin{proof}
    If $np^k\geq \log^3 n$, then $G_k$ has at most $2np^k$ vertices, with probability at least $1-\exp(-\log^2 n)$ by \Cref{claim:Gk vertices}.
    Then by \Cref{claim:z mm}, the running time is at most $\BO{\MM{n,n,{n/\Lambda^\gh}}}$.

    If $np^k\leq \log^3 n$, then $G_k$ has at most $\log^8 n$ vertices, with probability at least $1-\exp(-\log^2 n)$ by \Cref{claim:Gk vertices}.
    Then, the running time is at most $\BO{n}$:
    \begin{align*}
        \MM{\abs{V(G_k)},\abs{V(G_k)},\abs{V(G_k)}/\Lambda^\gh}
        \leq \MM{n^{1/3},n^{1/3},n^{1/3}}\leq n^{\omega/3}\leq n\;.
    \end{align*}
    which completes the proof.
\end{proof}
By \Cref{claim:f bound}, we conclude that the running time of each of the calls
$\set{\FindHalg(G_k,\Lambda_k)}_{k=0}^r$ is at most $r\cdot\BO{\MM{n,n,n/\Lambda^\gh}}$. As $r\leq \log n$, we get that all calls take a total of $\TO{\MM{n,n,n/\Lambda^\gh}}$ time.

Together with \Cref{conclusion:z1}, this completes the proof of \Cref{prop:rt bound}: All calls to the $\FindHalg$ black box and the $\Cialg$ black box take at most $\TO{{\MM{n,n,n/\Lambda^\gh}} + n^2/\eps^3}$ time in total.

\renewcommand{\Cialg}{{\mathsf{Count} \mhyphen\mathsf{Heavy}_{\eps}}}
\section{Implementing the Black Box $\Cialg$}\label{sec:count-heavy}
\label{sec:color code}
In this section, we implement the following black box.
\begin{blackbox}[${\Cialg(G,S,a,b)}$]%
    \begin{description}
        \item[Input:] A graph $G$, a precision parameter $\eps$, a subset of vertices $S$, and two real numbers $0<a\leq b$, such that
            $\forall v\in S$ we have $t_G(v)\in[a,b]$.
        \item[Output:]
            $\hat{t}$ which satisfies $\Pr{\hat{t}=t_G(S)(1\pm \eps)}\geq 1-\frac{1}{n^4}$.
    \end{description}
\end{blackbox}
In this section, we prove the second part of \Cref{lemma:assumption}, which is stated in the next theorem.
\begin{theorem}\label{thm:cialg main}
    \sloppy{There is an algorithm that implements the $\Cialg(G,S,a,b)$ black box, when $\FH$ is an $h$-cycles, for $h=\BO{1}$, in time $\TO{n^2\cdot \frac{b}{a\cdot\eps}}$.}
\end{theorem}

\newcommand{\tv}{t_{\varphi}(S)}
\newcommand{\htv}{Y}
\newcommand{\ACalgC}{\mathsf{Approx}_{t^k}}
\newcommand{\ACalgB}{\mathsf{Approx}_{\Exp{t^k}}}
\newcommand{\ACalgA}{\mathsf{Approx}_{\Exp{\mathsf{vertex}\mhyphen t^k}}}
For this entire section, the graph $G$ is fixed, and the set $S$ is a fixed subset of vertices, where for every $v\in S$ we have $t_G(v)\in[a,b]$.
We emphasize that $a$ is only a lower bound on $\min_{v\in S}t_G(v)$ and $b$ is only an upper bound on $\max_{v\in S}t_G(v)$.
Denote ${N_k\triangleq\abs{S}/k}$.
Before implementing the $\Cialg$ black box, we explain why the following simple idea fails: sample a few vertices from $S$, estimate $t_G(v)$ for each one, and apply a concentration bound to compute $t_G(S)$.
Formally, this approach fails because $\sum_{v\in S}t_G(v)\neq t_G(S)$.
To see why, note that the value $t_G(v)$ does not tell us how many copies of $\FH$ intersect $S$ more than once, and are therefore counted multiple times.
We present a simple example to demonstrate this.
Assume $\FH$ is a triangle, $\abs{S}=3$ and for every $v\in S$ we have $t_G(v)= 1$. Solely based on $t_G(v)$, one cannot distinguish between the cases  $t_G(S)=1$ and $t_G(S)=3$.

We overcome this issue by sampling a small subset of vertices $S'\subseteq S$, and then, for every $v\in S'$ we approximate the number of copies of $\FH$ which intersect $v$ and exactly $i$ additional vertices from $S$ for $0\leq i\leq h-1$. This will allow us to estimate the number of multiple countings and therefore get an estimation of $t_G(S)$.

\bigskip\noindent
Our approach for approximating $\hat{t}$ as required by $\Cialg$ is to approximate the number of cycles that intersect $S$ in exactly $k$ vertices. To this end, we define the following.

\begin{definition}
    Define $\FHs$ denote the set of copies of $\FH$ in $G$.
    For a subset of vertices $U$, we define $\FHs(U)\triangleq\set{C\in \FHs\mid C\cap U\nemp}$.
    We denote $t\triangleq\abs{\FHs}$ and $t(U)\triangleq\abs{\FHs(U)}$.
    In general, we replace the symbol $\FHs$ by $t$, to denote the cardinality of a set.
\end{definition}

\begin{definition}
    Let $\FHs^k$ denote the set of copies $C\in\FHs$ with $\abs{C\cap S}=k$.
    Let $\FHs^k(v)\triangleq \FHs^k\cap \FHs(v)$. We let $t^k=\abs{\FHs^k}$.
\end{definition}
This definition helps us overcome the double-counting issue because for any $k$ we have $\tfrac{1}{k}\sum_{v\in V}t^k(v)=t^k$, and $\sum_{k\in[h]}t^k=t_G(S)$. We now claim that we can efficiently approximate $t^k$.
\begin{lemma}[Algorithm $\ACalgC$]\label{lemma3:A4}
    There exists a randomized algorithm $\ACalgC$ with the following characteristics.
    The input is a graph $G$, a set $S$, a precision parameter $\delta\in[0,1]$, a parameter $k \in [h]$, and a tuple $(a,b)$, such that for every $v\in S$ we have $t_G(v)\in[a,b]$.
    The algorithm produces an output $\hat{t}^k$ which satisfies $\Pr{\hat{t}^k = t^k \pm t_G(S) \cdot \delta}\geq 1-\frac{1}{n^4}$.
    The running time of the algorithm is bounded by $\TO{n^2\cdot \frac{b}{a}\cdot \frac{1}{\delta}}$.
\end{lemma}

We now show how to prove \Cref{thm:cialg main} using \Cref{lemma3:A4}.
\begin{proof}[Proof of \Cref{thm:cialg main} using \Cref{lemma3:A4}]
    Set $\delta = \eps/2k$.
    Compute $\hat{t}^k$ for each $k\in[h]$ using \Cref{lemma3:A4}.
    Note that for each $k\in[h]$ we have
    $\Pr{\hat{t}^k=t^k\pm t_G(S)\cdot \delta} \geq\fpr{1}{4}$.
    Denote $\hat{t}_G(S)=\sum_{k\in[h]}\hat{t}^k$.
    We show that $\hat{t}_G(S)=t_G(S)\apm$, assuming $\hat{t}^k=t^k\pm t_G(S)\cdot \delta$ for every $k$, which happens with probability at least $\fpr{k}{4}
    $, by \Cref{lemma3:A4}.
    \begin{align*}
        \hat{t}_G(S)
        =\sum_{k\in[h]}\hat{t}^k
        =\sum_{k\in[h]} t^k\pm t_G(S)\cdot \delta
        =(\sum_{k\in[h]} t^k) \pm k\cdot t_G(S)\cdot \delta
        =t_G(S) \pm t_G(S)\cdot \frac{\eps}{2}
        =t_G(S)(1\pm \frac{\eps}{2}) \;.
    \end{align*}
    We make $h=\BC$ calls to the algorithm $\ACalgC$, so we get the desired running time, completing the proof of \Cref{thm:cialg main}.
\end{proof}

Thus, it remains to prove \Cref{lemma3:A4}.
To approximate $t^k$, we find a value whose expectation is $t^k$ and whose variance is at most $\BO{\brak{N_k\cdot b}^2}$ (recall that $N_k=\abs{S}/k$). The following lemma states that we can do this efficiently.
\begin{lemma}[Algorithm $\ACalgB$]\label{lemma3:A3}
    There is a randomized algorithm $\ACalgB$ whose input is $G,S,\delta$ and $k$. Note that unlike $\ACalgC$, the algorithm $\ACalgB$ does not require $a$ and $b$ as part of its input parameters.
    $\ACalgB$ computes a value $X$ such that
    $\Exp{X}=t^k$ and $\Var{X}\leq C\cdot \brak{N_k\cdot b}^2$, where $C$ is a constant.
    The running time of the algorithm is bounded by $\TO{n^2}$.
\end{lemma}

The reason that \Cref{lemma3:A3} is helpful is because we can run the algorithm it provides $r$ times and take the median of means of these invocations.
Formally, the proof of \Cref{lemma3:A4} is as follows.
\begin{proof}[Proof of \Cref{lemma3:A4} using \Cref{lemma3:A3}]
    Let $\ell\triangleq\frac{b}{a\cdot \delta}\cdot 3C$.
    Let $Y$ denote the mean of $\ell$ invocations of the algorithm $\ACalgB(G,S,\delta,k)$. Note that $Y$ is a random variable, where the probability space is over the internal choices of the algorithm.
    We claim that $\Pr{Y=t^k \pm t_G(S)\cdot \delta}\geq 2/3$.
    Assuming this is true, we can implement the algorithm $\ACalgC$ as follows.
    Make $\ell \cdot 400\log n$ calls to the algorithm $\ACalgB$.
    The algorithm $\ACalgC$ gets $a,b,$ and $\delta$ as inputs and therefore it can compute $\ell$.
    Split the outputs into batches of size $\ell$. Compute the mean of each batch, and output the median of means of the batches.
    Since the mean of each batch falls inside the desired interval with probability at least $2/3$, we can apply the ``median trick'' (see \Cref{claim:med trick0}) to show that the median of the means falls within this interval with probability at least $\fpr{1}{4}$.

    Therefore, to complete the proof, we are left with proving that $\Pr{Y=t^k \pm t_G(S)\cdot \delta}\geq 2/3$. We do so by
    bounding $\Var{Y}$ and using Chebyshev's inequality.
    We have
    \begin{align*}
        \Exp{Y}
        =\Exp{\frac{1}{\ell}\sum_{i\in \ell}\ACalgB(G,S,\delta,k)}
        =\Exp{\ACalgB(G,S,\delta,k)}=t^k\;,
    \end{align*}
    where the last equality is by the guarantees of $\ACalgB$ given in
    \Cref{lemma3:A3}.
    In addition,
    \begin{align*}
        \Var{Y}
        =\Var{\frac{1}{\ell}\sum_{i\in \ell}\ACalgB(G,S,\delta,k)}
        =\frac{1}{\ell^2}\Var{\ACalgB(G,S,\delta,k)}
        \leq \frac{C\cdot (N_k\cdot b)^2}{\ell^2}\;.
    \end{align*}
    Where the last equality is again by the guarantees of $\ACalgB$ given in \Cref{lemma3:A3}.
    Therefore, by Chebyshev's inequality we have
    \begin{align*}
        \Pr{\abs{Y-\Exp{Y}}\geq t_G(S)\cdot \delta}
        \leq \frac{\Var{Y}}{\brak{t_G(S)\cdot \delta}^2}
        \leq \frac{C\cdot (N_k\cdot b)^2}{{\brak{\ell\cdot t_G(S)\cdot \delta}^2}}
        = \frac{C\cdot (N_k\cdot b)^2}{{\brak{\frac{b}{a\delta}\cdot3 C \cdot t_G(S)\cdot \delta}^2}}
        \leq \frac{(a\cdot N_k)^2}{{\brak{3 \cdot t_G(S)}^2}}
        \leq \frac{1}{9}.
    \end{align*}
    The last inequality follows because $t_G(S) \geq a\cdot N_k$.
    To see that, note that each vertex in $S$ participates in at least $a$ copies of $\FH$.
    Moreover, each copy can contain at most $h$ vertices from $S$, therefore, the number of copies which intersect $S$ is at least $a\cdot \abs{S}/k$, which means that $t_G(S) \geq a\cdot \abs{S}/k=a\cdot N_k$.

    The running time is bounded by $\TO{n^2\cdot \frac{b}{a}\cdot \frac{1}{\delta}}$, as we made $\TO{\frac{b}{a\delta}}$ calls to the $\ACalgB$ algorithm, where each call takes $\TO{n^2}$ time by the guarantees of $\ACalgB$ given in \Cref{lemma3:A3}.
\end{proof}

In what follows, we prove \Cref{lemma3:A3}.
To get a sample $X$ with $\Exp{X}=t^k$ and $\Var{X}\leq C\cdot(b\cdot N_k)^2$, we find samples $Y_v$ with $\Exp{Y_v}=t^k(v)$ and $\Var{Y_v}\leq C\cdot b^2$.
The following lemma states that we can do this efficiently.
\begin{lemma}\label{lemma3:A2}
    There is a randomized algorithm $\ACalgA$ whose input is $G,S,\delta,k$ and a vertex $v\in S$.
    Unlike $\ACalgB$, $\ACalgA$ additionally takes a vertex $v \in S$ as input.
    $\ACalgA$ computes a value $Y_v$ such that
    $\Exp{Y_v}=t^k(v)$ and $\Var{Y_v}\leq C\cdot b^2$, where $C$ is a constant.
    The running time of the algorithm is bounded by $\TO{n^2}$.
\end{lemma}
The reason that \Cref{lemma3:A2} is helpful is that we can sample a vertex $v$ uniformly at random from $S$, and get, using \Cref{lemma3:A2}, an unbias estimator for the number of copies of $\FH$ which contains $v$, and intersect $S$ exactly $k$ times, i.e., $t^k(v)$. By the law of total expectation, we get that the expected value of this quantity, is equal to $\frac{1}{\abs{S}}\sum_{u\in S} \Exp{t^k(u)}=t^k/N_k$. Therefore, we got an unbias estimator for $t^k$ upto a known value $N_k$.
Formally, the proof of \Cref{lemma3:A3} is as follows.
\begin{proof}[Proof of \Cref{lemma3:A3} using \Cref{lemma3:A2}]
    \sloppy{For brevity, we write $\ACalgA(v)$, instead of $\ACalgA(G,S,\delta,k,v)$.
        We implement $\ACalgB$ using a single call to $\ACalgA$.
        The algorithm $\ACalgB$ samples a vertex $u\in S$ uniformly at random, and simulates $\ACalgA(u)$.
        Let $Y_u$ denote the output of the simulation.
        The algorithm $\ACalgB$ then outputs $X\gets Y_u\cdot N_k$.}
    We need to show that $\Exp{X}=t^k$ and $\Var{X}\leq\BO{(N_k b)^2}$, given that for all $u\in S$ we have $\Exp{Y_u}=t^k(u)$ and $\Var{Y_u}=C\cdot b^2$.
    For $u\in S$, define
    \begin{align*}
        Z_u\triangleq \begin{cases}
                          Y_u & \text{If $\ACalgB$ samples $u$} \\
                          0   & \text{Otherwise.}
                      \end{cases}
    \end{align*}
    Note that $X=N_k\cdot\sum_{u\in S}Z_u$.
    We compute the expectation and variance of $X$.
    \begin{align*}
        \Exp{X}
          & = N_k\sum_{u\in S}{\Exp{Z_u}}                                     \\
          & = N_k\sum_{u\in S}\Pr{u \text{ is sampled}}\cdot \Exp{\ACalgA(u)}
        = \frac{N_k}{\abs{S}}\sum_{u\in S}\Exp{\ACalgA(u)}                    \\
        &=  \frac1{k}\sum_{u\in S}t^k(u)
        = t^k\;.
    \end{align*}
    We now prove the claim on the variance.
    Recall that, by the guarantees of $\ACalgA$ given in \Cref{lemma3:A2}, we have that for every $u$, $\Var{Y_u}\leq C\cdot b^2$.
    Note that $\Var{Z_u}=\frac{\Var{Y_u}}{\abs{S}}\leq \frac{C\cdot b^2}{\abs{S}}$.
    Moreover, note that for any two distinct vertices $u,w\in S$ we have $\Exp{Z_u\cdot Z_w}=0$, as only one of them was chosen by $\ACalgB$.
    Therefore, $\mathsf{Cov}\brak{Z_u,Z_w}<0$.
    We get
    \begin{align*}
        \Var{X/N_k}
        = \Var{\sum_{u\in{S}}Z_u}
        = \sum_{u\in{S}}\Var{Z_u} + \sum_{u,w\in{S}}\mathsf{Cov}\brak{Z_u,Z_w}
        \leq \sum_{u\in{S}}\Var{Z_u}
        \leq Cb^2\;,
    \end{align*}
    and therefore $\Var{X}=N_k^2\cdot \Var{X/N_k}\leq C\cdot (N_k\cdot b)^2$.
\end{proof}
It remains to prove \Cref{lemma3:A2}, which we do in what follows.

\subsection{Implementing the algorithm $\ACalgA$}
We utilize the color-coding technique introduced by \cite{alon1995color}. The high level approach of the technique is to randomly color vertices with $h$ colors and detect \emph{colorful} $h$-cycles that are ordered by, say, increasing colors. This additional structure allows for faster detection, at the cost of some probability of missing $h$-cycles that are colored out of order. The technique overcomes this by repeated experiments.

A pertinent question arises: why are the algorithms $\ACalgB,\ACalgA$ necessary? Why not just choose a random coloring, compute the number of colorful  copies of $\FH$ intersecting a set $S$ exactly $k$ times, and apply Chebyshev's inequality to conclude that repeating this process $\TO{n^2 \cdot\frac{b}{a\delta}}$ times suffices for a good approximation of $t^k$? The answer
lies in the execution time of the matrix multiplication algorithm for counting colorful copies, which is dominated by the sizes of the largest, second largest, and smallest color classes. Roughly speaking, the smaller the product of these sizes, the faster the algorithm runs.
Under random coloring, color classes each have a size of $\Omega(n)$ with high probability. Conversely, $\ACalgA$ produces a color class containing just a single vertex, which significantly improves the running time in the worst case, compared to the approach which does not use the algorithms $\ACalgB,\ACalgA$.

~\\
We need the following definitions to explain how color-coding works, and how the algorithm $\ACalgA$ uses rectangular matrix multiplication.

\begin{definition}[$\varphi$-Colorful, $t_\varphi$]
    Fix some coloring $\varphi:V\to[\ell]$ for some $\ell\in\mathbb{N}$ ($\ell$ will usually be $h$).
    Let $\FHs_\varphi$ denote the set of all copies $C\in \FHs$, such that $\varphi(C)=[\ell]$. If $C\in \FHs_\varphi$, we say that $C$ is \emph{$\varphi$-colorful}.
    Also define $\FHs_\varphi(v)=\FHs_\varphi\cap \FHs(v)$,
    $t_\varphi\triangleq \abs{\FHs_\varphi}$, and $t_\varphi(v)\triangleq \abs{\FHs_\varphi(v)}$.
\end{definition}
The subsequent definition relates colorful copies of $\FH$ with the set $S$.
\begin{definition}[$\FHs^k_\varphi$]
    Let $\FHs^k_\varphi\triangleq \FHs^k\cap \FHs_\varphi$. That is, $\FHs^k_\varphi$ is the set of copies of $\FH$ in $G$, where each such copy is colorful w.r.t. $\varphi$, and additionally intersects the set $S$ exactly $k$ times. Let $t^k_\varphi=\abs{\FHs^k_\varphi}$.
\end{definition}
We use the term \emph{random coloring}, to denote a coloring of $V$ that is sampled at random from the uniform distribution. The following definition formally defines this term.
\begin{definition}[Random Coloring]
    Let $A,B$ be two finite sets.
    We say that a function $\varphi:A\to B$ is a random coloring, if the value of each $a\in A$ is set to some value $b\in B$, where $b$ is chosen uniformly at random from $B$ and independently of values chosen for other elements in $A$.
\end{definition}
The following definition is used to quantify the complexity of computing $t^k_\varphi$ as a function of the sizes of the color classes that the coloring $\varphi$ induces.
Let $\varphi:V\to[h]$.
For $i\in[h]$, we denote by $\varphi^{-1}(i)$ the set of all vertices $v$ with $\varphi(v)=i$, and call this set the $i$-th color class.
Assume without loss of generality that the color classes are sorted according their cardinalities, in a non-decreasing order. That is, for every $i< h$ we have $\varphi^{-1}(i)\geq \varphi^{-1}(i+1)$.

\begin{lemma}\label{prop:tk color}
    Let $(\threeclass)$ denote the cardinality of the largest, second largest, and smallest color classes, respectively.
    \sloppy{For any fixed $k\in [h]$, there is a deterministic algorithm for computing $\set{t^k_\varphi(v)}_{v\in V}$ in time $\BO{(h!)^2\cdot h^2\cdot \MM{\threeclass}}$.}
\end{lemma}
\begin{corollary}\label{cor:y}
    \sloppy{For any fixed $k\in [h]$, there is a deterministic algorithm for computing $t^k_\varphi$ in time $\BO{(h!)^2\cdot \MM{\threeclass}}$.}
\end{corollary}
The corollary follows because for every $k\in[h]$, we have $\tfrac{1}{k}\sum_{v\in V}t^k(v)=t^k$.
Next, we explain how to implement the algorithm $\ACalgA$ given that we can compute the number $t^k_\varphi$ of colorful copies of $\FH$.
The algorithm works as follows.
It colors each vertex with a random color from the set $[h-1]$.
It then recolors the input vertex by a new color $h$.
Let $\varphi$ denote this coloring. The algorithm then computes $t^k_\varphi$ and outputs $t^k_\varphi/q$ for some constant $q$ such that $\Exp{t_\varphi/q}=t^k(v)$. We prove that the expectation and variance of this output satisfy the claimed requirements.

\begin{proof}[Proof of \Cref{lemma3:A2} using \Cref{prop:tk color}]
    The algorithm $\ACalgA(G,S,\delta,k,v)$ is given a graph $G$, the set $S$, a precision parameter $\delta$, and a vertex $v\in S$.
    The algorithm samples a random coloring $\varphi':V\setminus \set{v}\to[h-1]$.
    Define a coloring $\varphi:V\to[h]$ as follows:
    \begin{align*}
        \forall u\in V\quad \varphi(u)\triangleq\begin{cases}
                                                    \varphi'(u) & \text{If } u\neq v \\
                                                    h           & \text{Otherwise}
                                                \end{cases}
    \end{align*}
    \noindent The algorithm then outputs $t_{\varphi}^k/q$, for $q\triangleq \frac{(h-1)!}{(h-1)^{h-1}}$. This completes the description of the algorithm $\ACalgA$.

    Next, we prove the algorithm's correctness assuming we can exactly compute $t_{\varphi}^k/q$.
    Denote the output of the algorithm by $Y_v$, where $Y_v$ is a random variable as $t_{\varphi}^k$ is also a random variable, where the randomness is over the choice of $\varphi$.
    We show that $\Exp{Y_v}=t^k(v)$ and $\Var{Y_v}=\BO{b^2}$.

    Consider some copy $C$ which is in $\FHs^k(v)$.
    The probability that $C\in \FHs^k_\varphi(v)$ is the probability that $C$ is colorful given that $C\in \FHs^k(v)$, which is equal to $q\triangleq \frac{(h-1)!}{(h-1)^{h-1}}$.

    Consider some copy $C$ which is not in $\FHs^k(v)$. We show that $C\notin \FHs^k_\varphi$. If $v\notin C$, then $C$ is not colorful, and therefore not in $\FHs^k_\varphi$. If $v\in C$ but $C\not\in \FHs^k(v)$, then $C\notin \FHs^k$ which means $C$ is also not in $\FHs^k_\varphi$.
    Therefore, we have
    \begin{align*}
        \Exp{Y_v}
        = \frac{1}{q}\sum_{C\in \FHs_G}\Pr{C \text{ is $\varphi$-colorful }}
        = \frac{1}{q}\sum_{C\in \FHs^k(v)}\Pr{C \text{ is $\varphi$-colorful }}
        = \frac{1}{q}\sum_{C\in \FHs^k(v)}q=t^k(v)\;.
    \end{align*}
    For the variance, we know that $t^k(v)\leq t_G(v)\leq b$ and that $Y_v\leq (b/q)$, which means that $(Y_v)^2\leq (b/q)^2$, and therefore $\Var{Y_v}\leq \Exp{(Y_v)^2}\leq (b/q)^2$.
    This completes the proof of the correctness of the algorithm $\ACalgA$.

    The running time of the algorithm $\ACalgA$ is bounded by the running time of the algorithm specified in \Cref{prop:tk color}, which is bounded by $\TO{\MM{n,n,1}}$ as the smallest color class has only one vertex, and other color classes has at most $n$ vertices. Note that $\TO{\MM{n,n,1}}= n^{2+o(1)}$, which proves the running time of the algorithm $\ACalgA$ is bounded by $\TO{n^2}$.
\end{proof}

We are left with proving \Cref{prop:tk color}, which is the final step in the implementation of the algorithm $\ACalgA$.
To prove \Cref{prop:tk color}, we reduce the problem of computing $t^k_\varphi$ to the problem of computing $t_\varphi$ on an auxiliary graph, in which every $h$-cycle is colorful and also intersects the set $S$ exactly $k$ times.
We construct the auxiliary graph by randomly coloring the vertices with $h$ colors, then selecting $k$ color classes and keeping only vertices from $S$ in them, while discarding the rest of the vertices in those classes. For the remaining $h-k$ color classes, we retain only vertices that are not part of $S$.
We get a graph in which each color class is either contained in $S$ or disjoint from $S$.
This reduces the problem of computing $t^k_\varphi$ on the auxiliary graph, to computing $t_\varphi$ on it.
The next claim addresses the running time of computing $t_\varphi$ (on the auxiliary graph) instead of computing $t^k_\varphi$.
\newcommand{\threeclassu}{U_{(1)},U_{(2)},U_{(h)}}
\begin{claim}\label{prop:t sigma}
    Let $G$ be a graph and
    let $\sigma\triangleq \brak{U_1,\ldots ,U_h}$ be an (ordered) sequence of disjoint subsets of vertices of $V(G)$.
    Let $t^\sigma_G$ denote the number of copies $C\in\FHs_G$ where $C=(v_1,v_2,\ldots, v_h)$ and $v_i\in U_i$ for every $i\in[h]$.
    Let $\threeclassu$ denote the cardinality of the largest, second largest, and smallest subset, respectively.
    Then, there is a deterministic algorithm that outputs $\set{t^{\sigma}_G(v)}_{v\in V}$ in time $\BO{h^2\cdot \MM{\threeclassu}}$.
\end{claim}

Next, we prove \Cref{prop:tk color} using \Cref{prop:t sigma}.
\begin{proof}[Proof of \Cref{prop:tk color} using \Cref{prop:t sigma}]
    \newcommand{\tvk}{t^k_\varphi(v)}
    \newcommand{\Fvk}{\FHs^k_\varphi(v)}
    We are given a graph $G$, a coloring $\varphi$ and a set $S$. Our goal is to compute $\set{t^k_\varphi(v)}_{v\in V}$ using \Cref{prop:t sigma}.
    First, we remove all monochromatic edges in $G$, as none of those edges participates in any $\varphi$-colorful copy $C\in \FHs_\varphi$.
    Then, we build a set of auxiliary graphs $\HH$.
    For each auxiliary graph $H_x\in \HH$, we compute $\set{t_{H_x}(v)}_{v\in V}$ using \Cref{prop:t sigma}. We then claim that $\sum_{H_x\in \HH}t_{H_x}(v)=t_\varphi^k(v)$. %
    The auxiliary graphs are defined as follows.
    For a string $x\in \set{0,1}^h$, we define a graph $H_x$, which is an induced subgraph of $G$.
    For $i\in [h]$, we define a set $U^x_i$ as follows:
    \begin{align*}
        U_i^x \triangleq
        \begin{cases}
            \varphi^{-1}(i)\cap S,      & \text{if } x_i = 1 \\
            \varphi^{-1}(i)\setminus S, & \text{if } x_i = 0
        \end{cases}
    \end{align*}
    Let $V(H_x)=\bigcup_{i\in [h]} U_i^x$ and $H_x=G[V(H_x)]$.
    Let $w(x)$ denote the hamming weight of a binary string.
    Let $\XC_\ell\triangleq\set{x\in \set{0,1}^h\mid w(x)=\ell}$.
    The set of auxiliary graphs is $\HH=\set{H_x\mid x\in \XC_k}$.
    Note that $\abs{\HH}\leq 2^h=\BO{1}$.
    This completes the description of the algorithm specified in \Cref{prop:tk color}. We proceed with proving its correctness, and then we analyze its running time.

    Fix some auxiliary graph $H_x\in \HH$.
    We write $\FHs_{\varphi}(v;H_x)$ to denote the number of $\varphi$-colorful copies in $H_x$ that $v$ participates in. We also use $t_{\varphi}(v;H_x)=\abs{\FHs_{\varphi}(v;H_x)}$.
    We first claim that it suffices to compute $\set{t_{\varphi}(v;H_x)}_{v\in V}$ for every $H_x\in \HH$, as
    \begin{align*}
        \tvk=\sum_{x\in\XC_k} t_{\varphi}(v;H_x)\;.
    \end{align*}
    To show that, we need to show two properties:
    \begin{enumerate*}[label=(\roman*)]
        \item every $C\in\bigcup_{x\in \XC_k}\FHs_{\varphi}(v;H_x)$ is also in $\Fvk$, and
        \item every $C\in\Fvk$ is in $\FHs_{\varphi}(v;H_x)$ for a unique $x$.
    \end{enumerate*}
    To see the first, note that every $\varphi$-colorful $C$ in $H_x$ must satisfy $\abs{C\cap S}=k$;
    If it is $\varphi$-colorful, it intersects each color class exactly once.
    By the construction of $H_x$, there are exactly $k$ color classes which are contained in $S$ while the other color classes are disjoint from $S$.
    Therefore, every copy $C$ on the right side of the equation, appears also on the left side of the equation. In other words, $\Fvk\supseteq\bigcup_{x\in \XC_k}\FHs_\varphi(v;H_x)$.

    We prove the second property.
    Fix $C\in \Fvk$, and denote its vertices by $(v_1,\ldots,v_h)$, where $\varphi(v_i)=i$ for every $i\in [h]$.
    To see that every $C$ is $\varphi$-colorful for at least one $x$,
    define $\hat{x}$ as follows. For $i\in[h]$, let $\hat{x_i}$ be $1$ if $v_i\in S$, and otherwise let $\hat{x_i}$ be $0$.
    $C$ is clearly $\varphi$-colorful in $H_{\hat{x}}$.
    To see that $C$ is not $\varphi$-colorful in $H_{x'}$ for any $x'\neq \hat{x}$, let $j$ be the first index in which $x'\neq \hat x$. If $C$ is also $\varphi$-colorful in $H_{x'}$, then $v_j$ is both in $S$ and not in $S$ which is a contradiction.

    We are left with explaining how to compute $\set{t_{\varphi}(v;H_x)}_{v\in V}$ for every $H_x\in \HH$, which we do by using \Cref{prop:t sigma}.
    Let $\Pi_h$ denote the set of all permutations $\pi:[h]\to[h]$.
    For each permutation $\pi\in \Pi_h$, define an ordering $\sigma(\pi)=\brak{\varphi^{-1}(\pi(1)),\ldots, \varphi^{-1}(\pi(h))}$, and compute $\set{t_{H_x}^{\sigma(\pi)}(v)}_{v\in V}$ using \Cref{prop:t sigma}.
    To complete the proof, we need to show that for every $v$, and every $x\in\XC_k$ we have
    \begin{align*}
        2h\cdot t_{\varphi}(v;H_x) =\sum_{\pi\in \Pi_h}t_{H_x}^{\sigma(\pi)}(v)\;.
    \end{align*}
    Fix $H_x$, a vertex $v$, and a cycle $C = (v_1,\ldots,v_h) \in \FHs_{\varphi}(v;H_x)$, with $\varphi(v_i) = i$ for all $i \in [h]$.
    We show that there are exactly $2h$ permutations in $\Pi_h$ for which $C\in \FHs_{H_x}^{\sigma(\pi)}(v)$. Take any permutation $\pi\in \Pi_h$. It defines a permutation $\pi':V(C)\to V(C)$ as follows. The vertex colored by the $i$-th color, is mapped to the vertex colored by the $\pi(i)$-th color.
    $C$ will be in $\FHs_{H_x}^{\sigma(\pi)}(v)$ if and only if $\pi'$ is an isomorphism on $C$. That is, for every edge $(u,v)$ in $C$, the edge $(\pi
        '(u),\pi'(v))$ is also an edge in $C$. This set of automorphism of an $h$-cycle has $2h$ elements, which consist of a \emph{shift} permutations, and its composition with itself $i$ times for any $1\leq i\leq h$, and a \emph{reflection} permutation, which maps the vertex in the $i$-th position to the $h+1-i$-th position. The composition of each of the shift permutation with a reflection permutation forms another set of $h$ different permutation, making the total number of permutation that belong to the set of automorphism of $C$ equal to exactly $2h$.
    This completes the correctness of the lemma.
    \paragraph{Running Time.}
    For a given $x\in \XC_k$ and a specific permutation $\pi\in \Pi_h$, the computation of $\set{t_{H_x}^{\sigma(\pi)}(v)}_{v\in V}$ as per \Cref{prop:t sigma} requires $\BO{h^2\cdot \MM{\threeclassu}}$ time. Consequently, determining $\set{t_{\varphi}(v;H_x)}_{v\in V}$ demands $\BO{h! \cdot h^2\cdot \MM{\threeclassu}}$ time as $\abs{\Pi_h}=h$.
    Extending this to all $x\in \XC_k$ leads to a computational time of $\BO{(h!)^2\cdot h^2\cdot \MM{\threeclassu}}$, as $\abs{\XC_k}\leq h!$.
    Thus, the computation of $\set{\tvk}_{v\in V}$ also requires $\BO{(h!)^2\cdot h^2\cdot \MM{\threeclassu}}$ time. The proof concludes with the observation that $\MM{\threeclassu} \leq \MM{\threeclass}$.
\end{proof}

We now prove \Cref{prop:t sigma}, which is the final missing piece for the proof of \Cref{thm:cialg main}.
\begin{proof}[Proof of \Cref{prop:t sigma}]
    We assume without loss of generality that $\abs{U_1}=U_{(h)}$, i.e., that $U_1$ is the smallest set. %
    We first explain how to compute $\set{t^{\sigma}_G(v)}_{v\in U_1}$, and then we generalize this for $U_j$ for any $j\in[h]$.

    Let $A$ denote the adjacency matrix of $G$.
    For $X,Y\subseteq V(G)$ let $A[X,Y]$ denote the submatrix containing all rows $v$ for $v\in X$ and all columns $u$ for $u\in Y$.
    Combinatorially, define a new directed graph $H'$ with vertex set $X\cup Y$, and a directed edge $(x,y)$ between a pair of vertices $x\in X$ and $y\in Y$ if and only if $(x,y)$ is an edge in $G$. Note that $A[X,Y]$ is exactly the adjacency matrix of the new graph $H'$. %
    Define $B_0\triangleq I_{\abs{U_1}}$, and for $0\leq i< h$,
    define
    \begin{align*}
        A^\sigma_i\triangleq A[U_i,U_{i+1}]\;, &  &
        B^\sigma_{i}\triangleq B_{i-1}\cdot A_{i}
        \;.
    \end{align*}
    We compute $B^\sigma_i$ for $0\leq i<h$.
    Note that $B^\sigma_i[x,y]$ denotes the number of paths with $i$ edges between a vertex $x\in U_1$, and a vertex $y\in U_{i+1}$, %
    which are of the form $(x,u_2,u_3,\ldots, u_{i-1},y)$ where $u_j\in U_j$ for $2\leq j< i$.
    Let $A_h\triangleq A[U_h,U_1]$.
    After computing $B_{h}$, we compute $M = B^\sigma_{h}\cdot A^\sigma_h$ and return all entries on the diagonal of $M$. Note that for $v\in U_1$, we have that $M[v,v]=t^{\sigma}_G(v)$.

    Fix some $j\in[h]$. We explain how to compute $\set{t^\sigma_G(v)}_{v\in U_j}$, when $j\neq 1$.
    We apply the same algorithm as for $j=1$, with a small modification.
    We compute two matrices.
    The first matrix, denoted by $M_{j,1}$, satisfies that for all $x\in U_1$, and $y\in U_j$ we have that $M_{j,1}[x,y]$ is equal to the number of paths with $j-1$ edges between $x$ and $y$, which are of the form $(x,u_2,u_3,\ldots, u_{j-1},y)$ where $u_k\in U_k$ for $2\leq k< j$.
    The second matrix, denoted by $M_{j,2}$, satisfies that for all $x\in U_1$, and $y\in U_j$ we have that $M_{j,2}[x,y]$ is equal to the number of paths with $h-(i-1)$ edges between $x$ and $y$, which are of the form $(x,u_h,u_{h-1},\ldots, u_{j+2}, u_{j+1},y)$ where $u_k\in U_k$ for $j< k\leq h$.
    We then compute $D=M_{j,1}^{T}\cdot M_{j,2}^{T}$ where $M^T$ denotes the transposed matrix of $M$.
    Note that $M_{j,1}=B^\sigma_j$, so we can simply compute $B^\sigma_j$ as before.
    We then define $\sigma'\triangleq (U_1,U_h,U_{h-1},\ldots, U_2)$, and compute $B^{\sigma'}_{j}$
    where $j=h-i$. Then, set $M_{j,2}=B^{\sigma'}_{j}$.
    This completes the description of the algorithm.

    \paragraph{Running Time.}
    In the $i$-th iteration, for $1\leq i\leq h$, we compute the product of the matrices $B_{i-1}$ with the matrix $A_{i}$.
    We denote the dimensions of $A_{i}$ by $a',b'$.
    The dimensions of $B_{i-1}$ are ${U_{(h)},a'}$, and therefore the running time is $\MM{a',b',U_{(h)}}$.
    Without loss of generality, we can assume $b'\leq a'$, because $\MM{a',b',X}=\MM{X,b',a'}$ for any $a',b',X$. We also have that $a'\leq U_{(1)}$, and that $b'\leq U_{(2)}$.
    This bounds the running time of the $i$-th iteration by $\BO{h^2\cdot \MM{a',b',U_{(h)}}}\leq \BO{h^2\cdot \MM{\threeclassu}}$, which proves the claim.
\end{proof}

\section{Implementing the $\FindHalg$ Black Box}\label{sec:Find Heavy}
\renewcommand{\tv}{t_\varphi(v)}
\newcommand{\Tv}{\FHs_\varphi(v)}
\newcommand{\Tvv}[1]{\FHs_\varphi(#1)}
\newcommand{\Gpi}{G_\varphi(\pi)}
\newcommand{\Gc}{G_\varphi}
\newcommand{\Gpc}{G_\varphi}
\newcommand{\tvc}{t_{\Gpc}(v)}
\newcommand{\Gpiv}{G_\varphi(\pi_v)}
\newcommand{\ct}[1][h]{(2h\log(n))^{#1}}
\newcommand{\ctz}[1][h]{(2h\log(n))^{(#1)^2}}
\newcommand{\phval}[1][h-1]{(\eo)^{#1}}
\newcommand{\phvalq}[1][h-1]{(\eo)^{#1}\cdot \frac{q}{2}}
\newcommand{\sq}{s}
In this section, we prove the first part of \Cref{lemma:assumption}, as is specified in the following theorem.
\begin{theorem}\label{thm4:find heavy}
    There is an algorithm that implements the $\FindHalg{(G,\Lambda)}$ black box when $\FH$ is an $h$-cycle with $h=\BO{1}$, in time $\TO{\MM{n,n,n/\Lambda^{1/(h-2)}}}$.
\end{theorem}
An algorithm for \Cref{thm4:find heavy} is given a graph $G$ and a heaviness threshold $\Lambda$, and needs to output a superset of the $\Lambda$-heavy vertices, which contains no $\Lambda/\clog$-light vertices.
Our algorithm works as follows. The algorithm selects a vector
$P=(p_1,\ldots, p_h)$, where $p_i\in [0,1]$ for $i\in[h]$, which we explain shortly how to select. The algorithm then samples a uniform coloring $\varphi$ for the vertices, and keeps each vertex of the $i$-th color class with probability $p_i$.
We emphasize that not all vertices are kept with the same probability.
Let $F$ denote the obtained graph. %
If $v$ is a part of at least one $\varphi$-colorful copy of $\FH$ in $F$, we say that $v$ is $P$-\emph{discovered}. We can find all $P$-discovered vertices over $F$ using \Cref{prop:tk color}.
For every vertex $v\in V$, let $v[P]$ denote the probability that $v$ is
$P$-discovered.
The randomness is taken over the choice of the coloring and the sampling of vertices. We call this experiment the \emph{$P$-discovery experiment}.

We repeat this $P$-discovery experiment $k$ times. If $v$ is $P$-discovered more than $k\tau$ times, where $\tau$ is a threshold that we set later, then $P$ adds $v$ to the set of heavy vertices. In this case, we say that $v$ is \emph{$P$-added} to the set of heavy vertices.

The final step of our algorithm is to choose a vector $P$, or rather a set of vectors $\mathcal{P}$, and then to perform the $P$-discovery experiment with each vector in the set $k$ times. Before specifying how we should choose $\mathcal{P}$, $k$, and $\tau$, we state the properties we hope to achieve.
\begin{enumerate}
    \item Each vector $P\in\mathcal{P}$ induces a graph $F$, such that invoking \Cref{prop:tk color} for computing the set of $P$-discovered vertices, i.e., vertices that are part of $\varphi$-colorful copy of $\FH$, takes $\TO{\MM{n,n,n/\Lambda^{1/(h-2)}}}$ time.
    \item Each vertex with $t_G(v)\geq \Lambda$ has at least one vector $P\in\mathcal{P}$ that $P$-adds $v$ to the set of heavy vertices with high probability.
    \item For any vertex $v$ with $t_G(v)\leq \Lambda/\clog$ and any vector $P\in \mathcal{P}$, with high probability, $v$ is not $P$-added to the heavy vertex set.
    \item The set $\mathcal{P}$ has only $\TO{1}$ vectors, allowing us to avoid repeating this experiment too many times.
\end{enumerate}
Next, we give an explicit definition of the set of all vectors that our algorithm uses. This set is built in the following way. Take all the vectors $(p_1,\ldots, p_h)\in [0,1]^h$ that satisfy $\prod_{i\in[h]}p_i\leq \TO{\frac{1}{\Lambda}}$.  We take a small (finite) subset of this set, in which every vector in the original set is close to some vector in the reduced set.
To do that, we only keep vectors $(p_1,\ldots, p_h)\in [0,1]^h$ for which $p_i\in\set{2^{-j}\mid 0\leq j\leq \log(\Lambda) + 1}$ for $i\in[h]$.
\begin{definition}[Vector Set: $\Prod(\Lambda)$]
    For a non-negative real $\Lambda$, define the following subsets of $[0,1]^h$.
    \begin{align*}
        \mathcal{A}(\Lambda) & \triangleq\set{(p_1,\ldots,p_h)\in[0,1]^h\mid \forall i\in[h]\;,\quad p_i\in\set{2^{-j}\mid 0\leq j\leq \log(\Lambda)+1}} \\
        \mathcal{B}(\Lambda) & \triangleq\set{(p_1,\ldots,p_h)\in[0,1]^h\mid \prod_{i=1}^h p_i\leq \frac{1}{\Lambda}}                                    \\
        \Prod(\Lambda)       & \triangleq \mathcal{A}(\Lambda)\cap \mathcal{B}(\Lambda)
    \end{align*}
\end{definition}
Note that $\abs{\Prod(\Lambda)}\leq \abs{\AC(\Lambda)}\leq \brak{\log(\Lambda)+2}^h=\TO{1}$, where the last inequality follows because $\Lambda\leq n^h$, since no vertex in $G$ participates in more than $n^h$ copies of any $h$-vertex graph (that is, if the input was $\Lambda> n^h$, the algorithm could simply output an empty set). This means that $\log(\Lambda)\leq h\log n=\TO{1}$. This proves   Property (4) above.

The rest of the section is dedicated to proving that this set satisfies Properties (1)--(3).
We first prove the first property, which states that for each $P\in \Prd$, the $P$-discovery experiment %
can be implemented in the desired time. %
\begin{lemma}\label{lemma:simulate query}
    Let $G$ be a graph with $n$ vertices and let $\Lambda$ be some positive number.
    Let $P=(p_1,\ldots, p_h)$ be a vector in $[0,1]^h$ with
    $\prod_{i=1}^{h} p_i= \TO{\frac{1}{\Lambda}}$.
    Let $F$ be the (random) graph obtained in a $P$-discovery experiment.
    Then, we can find all the $P$-discovered vertices in time $\TO{\MM{n,n,n\cdot \Lambda^{-1/(h-2)}}}$, \whp.
\end{lemma}
\begin{proof}[Proof of \Cref{lemma:simulate query}]
    Fix some $P\in [0,1]^{h}$ where $P=(p_1,\ldots,p_{h})$, and
    $\prod_{i=1}^{h} p_i\leq 1/X$, for $X\geq 1$.
    Assume without loss of generality that $p_{i}\geq p_{j}$ for $i > j$. This implies that $p_1$ is the largest coordinate.
    Let $F$ be the random graph, and let $F_1,F_2,F_h$ denote the first, second, and $h$-th color classes in $F$.
    First, note that we can use \Cref{prop:tk color} to computing the set of discovered vertices in time $\BO{\MM{\abs{F_1},\abs{F_2},\abs{F_h}}}$ (recall that $h$ is constant).
    This follows because for each vertex, we get the number of colorful copies over $F$ in which it participates, and therefore we learn which vertices are $P$-discovered. For the rest of the proof, we analyze the running time of the algorithm specified in \Cref{prop:tk color} on the random graph $F$.
    Using a standard Chernoff's inequality, we claim that $\abs{F_i}\leq 30\log n \cdot \max\set{\Exp{\abs{F_i}},1}$. This means that the running time is $\TO{\MM{np_1,np_2,np_h}}$ with high probability.
    We show that
    \begin{align}
        \label{eq:MM}
        \MM{np_1,np_2,np_h}
        \leq\MM{n,n,n\cdot (p_1p_2p_h)}
        \leq\MM{n,n,n/X^{1/(h-2)}}\;,
    \end{align}
    and complete the proof by setting $X\gets \TO{\Lambda}$.
    The first inequality in \Cref{eq:MM} follows by \Cref{claim:a2}.
    We prove the second inequality, by showing that $p_1\cdot p_2\cdot p_h\leq X^{-1/(h-2)}$.
    Define $a_i$ where $p_i=2^{-a_i}$ for $i\in[h]$.
    Note that while $p_i\geq p_{i+1}$, we have $a_i\leq a_{i+1}$.
    We need to show that $a_1+a_2+a_h\geq \frac{\log(X)}{h-2}$.
    To minimize $a_1+a_2+a_h$ while keeping $\sum_{i\in[h]} a_i \geq \log(X)$,
    we set $a_1=a_2=0$ and $a_i= \frac{\log(X)}{h-2}$ for $2< i\leq h$. Call this solution $R$. %
    Note that $a_1+a_2+a_h=\frac{\log(X)}{h-2}$ and therefore $p_1\cdot p_2\cdot p_h=X^{-1/(h-2)}$.
    We show that $R$ is a solution which obtains the minimal value for the function $a_1+a_2+a_h$ under the described constraint on $a_i$.
    To see this, take another solution $R'=\set{a'_i}_{i\in[h]}$ and look at the value of $a_h'$ in this solution.
    If $a_h'>a_h$ then the value of the solution $R'$ is larger than the value of $R$ because $a_1=a_2=0$.
    The other case is when $a_h'\leq a_h$, for which we denote $a_h'=a_h - \delta$ for some positive $\delta$.
    In this case, since $a_h'$ is the largest value among all $a'_i$,
    to maintain the sum constraint, the sum $a'_1+a'_2$ must satisfy $a'_1+a'_2\geq a_1 + a_2 + \delta(h-2)$.
    Therefore,
    \begin{align*}
        a'_1+a'_2+a'_h \geq a_1+a_2 +\delta(h-2) +a_h -\delta \geq a_1+a_2+a_h,
    \end{align*}
    where the last inequality is because $h\geq 3$.
    This completes the proof of the second inequality in \Cref{eq:MM}.
\end{proof}

The main part in this section is proving Properties $(2)$ and $(3)$ which we do next.

\begin{theorem}\label{thm5:square main2}
    Let $G$ be a graph and let $\sq\triangleq\frac{1}{h^h}$. For any $\Lambda$ where $\Lambda=\omega(1)$, we have
    \begin{enumerate}
        \item For every vertex $v$ with $t_G(v)\geq \Lambda\cdot \ctz[h-1]\cdot\frac2{\sq}$, there exists a vector $P\in \Prod(\Lambda)$, such that $v[P]\geq \phval \cdot \sq$.
        \item For every vertex $v$ with $t_G(v)\leq \Lambda/\log n$, and every vector $P\in \Prod(\Lambda)$, we have $v[P]\leq 1/\log n$.
    \end{enumerate}
\end{theorem}
We now prove \Cref{thm4:find heavy}.
\begin{proof}[Proof of \Cref{thm4:find heavy} Using \Cref{thm5:square main2}]
    We first deal with the case that $\Lambda=\TO{1}$.
    In this case, we need to find a set $S$ which contains all vertices $v$ with $t_G(v)>0$, and no vertex $u$ with $t_G(u)=0$ with probability at least $\fpr{1}{4}$.
    In this case we do not sample subsets of vertices, but rather we only sample a random coloring and find the set of vertices that participate in at least one colorful $h$-cycle.
    We add each of these vertices to the set $S$ and repeat this process with a new sampled random coloring.
    Obviously, no vertex $u$ with $t_G(u)=0$ gets added to $S$.
    We show that after running this algorithm $k=\kVal$ times, each vertex with $t_G(v)>0$ is in $S$ with probability at least $1-\frac{1}{\exp(\Omc[k])}$.
    Fix some vertex $v$ with $t_G(v)>0$. To show that $v$ is added to $S$,
    we need to compute the probability that a copy $C\in \FHs(v)$  becomes colorful under random coloring.
    This probability is exactly $q\triangleq \frac{h!}{h^h}$. Note that $1/q=\BO{1}$. By repeating this process $k$ times, we get that with probability at least $1-\exp(-kq)\geq \fpr{1}{6}$, $v$ is added to $S$.
    By a union bound over all vertices, we get that $S$ is valid with probability at least $\fpr{1}{5}$. This completes the proof for the case $\Lambda=\TO{1}$.
    \medskip
    \newcommand*{\tl}{\tilde{\Lambda}}

    We move on to case where $\Lambda=\omega(1)$.
    We want to distinguish between $\Lambda$-heavy and $\Lambda/\czlog$-light vertices. Let $\tl\triangleq\Lambda/\brak{\ctz[h-1]\cdot\frac2{s}}$.
    We apply \Cref{thm5:square main2} with $\tl$ to distinguish between
    $\Lambda_1$-heavy and $\Lambda_2$-light vertices, where
    $\Lambda_1=\Lambda = \tl\cdot \ctz[h-1]\cdot\frac2{s}$, and $\Lambda_2=\tl/\log n \geq \Lambda/\czlog$, using vectors in $\Prod(\tl)$.

    The algorithm for \Cref{thm4:find heavy} is as follows.
    For each vector $P\in \Prod(\tl)$, run the $P$-discovery experiment $k$ times, where $k=\kVal$.
    Let $\tau \triangleq k\cdot \phval\cdot \frac{s}{4}$.
    Let $S_P$ denote the set of vertices that are $P$-discovered more than $k\tau$ times and let $S=\bigcup_{P\in\Prd}S_P$.
    The algorithm outputs $S$ as the superset containing all $\Lambda_1$-heavy vertices and no $\Lambda_2$-light vertices.
    This completes the algorithm description.
    We set $\tau$ to (at least) half the expected number of times a heavy vertex $v$ is $P$-discovered by some $P$, while also keeping $\tau$ large enough such that $\exp(-\tau/10)\leq \frac{1}{n^6}$.
    We prove that $S$ is indeed a valid output \whp.

    Fix a $\Lambda_2$-light vertex $v$, and a vector $P\in\Prod(\tl)$.
    By \Cref{thm5:square main2} we have that $v[P]\leq 1/\log n$.
    We show that $v\not\in S_P$ with probability at least $1/n^6$. We then use a union bound over all $\Lambda_2$-light vertices and vectors, and get that $S$ does not contain any $\Lambda_2$-light vertices with probability at least $1-\TO{1}/n^5$.
    The probability that $v\in S_P$ is equal the probability that a Binomial random variable $X_v$ with $k$ trials, and success probability $v[P]$, obtains a value of at least $\tau$. Note that $\Exp{X_v}=k\cdot v[P]\leq \tau/\log n$, and therefore, by Chernoff's inequality (\Cref{thm:chernoff 6}), we get that $v\not\in S_P$ with probability at least $1-2^{-\tau}\leq \fpr{1}{6}$.

    Fix a $\Lambda_1$-heavy vertex $v$, and a vector $P_v\in\Prod(\tl)$, where $v[P_v]\geq \phval\cdot s$, where such a vector exists by \Cref{thm5:square main2}.
    We show that $v\in S_{P_v}$ with probability at least $1-1/n^6$. We then use a union bound over all $\Lambda_1$-vertices, and get that $S$ contains all $\Lambda_1$-hevay vertices with probability at least $\fpr{1}{5}$.
    The probability that $v\in S_P$ is equal the probability that a Binomial random variable $X_v$ with $k$ trials, and success probability $v[P_v]$, obtains a value of at least $\tau$. Note that $\Exp{X_v}=k\cdot v[P]\geq  4\tau$, and therefore, by Chernoff's inequality (\Cref{thm:chernoff 6}), we get that $v\in S_P$ with probability at least $1-\exp\brak{-2\tau/8}\leq \fpr{1}{6}$.
    Therefore, $S$ is a valid output with probability at least $\fpr{1}{4}$.

    To analyze the running time, we use \Cref{lemma:simulate query}, which states that for each $P\in\Prod(\tl)$, executing the $P$-discovery experiment $k = \log^4 n$ times takes
    $\TO{\MM{n,n,n\cdot \tl^{-1/(h-2)}}}=\TO{\MM{n,n,n\cdot \Lambda^{-1/(h-2)}}}$ time, \whp.
    Since $\abs{\Prod(\tl)}=\TO{1}$, we get total running time of $\TO{\MM{n,n,n\cdot \Lambda^{-1/(h-2)}}}$, \whp.
\end{proof}
For the rest of this section, we prove \Cref{thm5:square main2}, under the assumption that $\Lambda=\omega(1)$.
We begin by providing an upper bound on the probability that a light vertex is  $P$-discovered for any vector $P\in \Prd$.
\begin{proposition}
    Let $Y$ be any value that might depend on $n,h$ and $\Lambda$.
    Fix a vertex $v$ with $t_G(v)\leq \Lambda/Y$ and some vector $P\in \Prod(\Lambda)$. Then, $v[P]\leq \frac{q}{Y}$, where $q\triangleq\frac{h!}{h^h}$.
\end{proposition}
\begin{proof}
    Fix some vector $P\in\Prod(\Lambda)$ where $P\triangleq (p_i)_{i=1}^{h}$.
    Sample a random $h$-coloring $\varphi$.
    For every $C\in \FHs_G(v)$ let $X_C$ denote the random variable indicating whether $C$ is also $\varphi$-colorful in the random graph $F$ that is obtained by sampling according to $P$. The randomness is over the choice of $\varphi$, as well as the vertices sampled from each class.
    Fix some copy $C_0\in \FHs_G(v)$.
    The probability that $C_0$ is $\varphi$-colorful in $G$ is exactly $q=\frac{h!}{h^h}$.
    The probability that $C_0$ is $\varphi$-colorful in $F$ is therefore
    \begin{align*}
        \Pr{C_0 \text{ is $\varphi$-colorful in }F}=\Exp{X_{C_0}}=\Pr{C_0\in\FHs_\varphi}\cdot \prod_{i=1}^h p_i\leq q\cdot \frac{1}{\Lambda}\;,
    \end{align*}
    where the last inequality follows by the definition of $\Prod(\Lambda)$.
    Let $X$ denote the random variable equal to the number of copies $C\in \FHs_\varphi(v)$ which are also $\varphi$-colorful in $F$.
    That is, $X=\sum_{C\in \FHs_G(v)}X_C$. Note that $v$ is $P$-discovered if and only if $X>0$.
    We have
    \begin{align*}
        \Exp{X}=\sum_{C\in \FHs_G(v)}\Exp{X_C}\leq t_G(v)\cdot \Exp{X_{C_0}}
        \leq \frac{\Lambda}{Y}\cdot \frac{q}{\Lambda}= \frac{q}{Y}\;.
    \end{align*}
    Therefore, by Markov's inequality:
    $\Pr{X>0}= \Pr{X\geq1}\leq \Exp{X}\leq \frac{q}{Y}$,
    where the first equality follows because $X$ takes integral values.
    Since $v$ is $P$-discovered if and only if $X>0$, we get that $v[P]$ is bounded by $\frac{q}{Y}$, which completes the proof.
\end{proof}
For the rest of this section we show that any heavy vertex $v$ has a vector $P$ such that $v[P]$ is sufficiently large.
\begin{proposition}\label{prop:heavy main}
    Fix a vertex $v$ with $t_G(v)\geq \Lambda\cdot \ctz[h-1]\cdot 2/\sq$, where $\sq\triangleq \frac{1}{h^h}$.
    Then, there exists a vector $P\in \Prod(\Lambda)$ such that $v[P]\geq \phval\cdot s$.
\end{proposition}

We prove the proposition by induction on $h$, for which the base case is $h=3$.

The intuition for the proof is as follows. Consider the base case of triangles. Let $V_i$ denote the $i$-th color class.
Consider finding the heavy vertices in the first color class.
Consider all vectors $P_i=(1,2^{-i},\frac{2^i}{\Lambda})$, for $0\leq i\leq\log(\Lambda) + 1$. %
These vectors form a subset of $\Prd$.
Fix some $\Lambda$-heavy vertex $v\in V_1$.
We want to prove that for at least one $i\in[\log(\Lambda)]$ we have $v[P_i]\geq \frac{q(\eo)^2}{2}$.
Our choice for the vectors is designed to deal with the following two extreme cases.
The first case is that every $C\in\FHs_\varphi(v)$ intersects one specific vertex $u\in V_2$. For this case, the vector $P_0 = (1,1,\frac{1}{\Lambda})$ is the right choice for discovering $v$ because it maximizes the probability we hit $u$ and a common neighbor of $v$ and $u$, under the constraint that the sampling probability $p_1\cdot p_2\cdot p_3\leq 1/\Lambda$.
The second case is that among each of $V_2\cap N(v)$ and $V_3\cap N(v)$, there are $\sqrt{\Lambda}$ vertices that
connected as a complete bipartite graph (contained in $V_2\times V_3$).
For this case, the vector $P_{i}$ for $i=\log(\Lambda)/2$ is the right choice for discovering $v$, because $P_i = (1,\frac{1}{\sqrt{\Lambda}},\frac{1}{\sqrt{\Lambda}})$.

For a larger $h$, the ``hard'' case is the following generalization of the above. From each color class $V_i$ for $i\in\{3,\dots,h\}$ we take $k=\Lambda^{1/(h-2)}$ vertices and connect them by a complete $(h-2)$-partite graph. We then take $v\in V_1$ and $u\in V_2$, we add an edge between $v$ and $u$, and we connect $v$ to all the aforementioned $k$ vertices from $V_{h}$ and we connect $u$ to all the aforementioned $k$ vertices from $V_{3}$. The vector $(1,1,\frac{1}{k},\dots,\frac{1}{k})$ is the right choice for this case.

Roughly speaking, we show the set $\Prd$ gradually shifts from handling one extreme case to another, and therefore ``covers'' all cases in between, which are the different ways to split vertices in $h$-cycles into $h$ pieces whose product of sizes is $\Lambda$.

We need a final technical step before presenting the proof.
First, we construct an $h$-partite graph which is easier to work with.
For this, we sample uniform $h$ coloring of the vertices of $G$. For $i\in [h]$ let $V_i$ denote the set of vertices colored in the $i$-th color.
For every $i\in[h]$ we keep only edges between the vertices of $V_i$ and $V_{i+1 \mod h}$, and direct those edges from $V_i$ to $V_{i+1}$.
Let $\Gc$ denote the obtained directed graph.
We emphasize that $t_{\Gc}(v)\leq \tv$, as the latter counts all colorful cycles in which $v$ participates in, whereas the former counts only colorful cycles with edges between $V_i$ and $V_{i+1}$ in which $v$ participates in.
We prove that a heavy vertex will be $P$-discovered on $\Gpc$, with probability at least $\Omc$, by some $P\in\Prd$, whereas light vertices will only be discovered with probability $\BO{1/\log n}$. In other words, since the gap between the heavy and light vertices is sufficiently big, we can still distinguish between the two in $\Gpc$.

The rest of this section is dedicated to proving the following proposition
\begin{lemma}\label{prop:discover over Gpi}
    Fix a vertex $v$, and an $h$-coloring of the vertices of $G$.
    Suppose $\tvc\geq \Lambda\cdot \ctz[h-1]$.
    Then, there exists a vector $P\in \Prod(\Lambda)$ such that the probability that $v$ gets $P$-discovered in $\Gpc$ is at least $\phval$.
\end{lemma}
\begin{proof}[Proof of \Cref{prop:discover over Gpi}]
    We prove \Cref{prop:discover over Gpi} by induction on $h$.
    We start with the base case, i.e., $h=3$.
    We fix some $h$-coloring $\varphi$.
    Recall that we denote the $i$-th color class by $V_i$.
    Consider a vertex $v\in V_1$ with $\tvc\geq \Lambda\cdot \ct[4]$.
    To simplify the notation,
    we use $G'=\Gpc$.
    We will prove that there exists a vector $P\in \Prd$, such that the probability that the vertex $v$ is $P$-discovered over $G'$ is at least $\phval=\phval[2]$.
    Let $K_0=\log(\Lambda)+1$.
    We consider a subset of vectors in $\Prd$ of the form $P_i=(1,p_i,q_i)$, for $i\in\zrn{K_0}$, where $p_{i}\triangleq 2^{-i}\;, q_{i}\triangleq \tfrac{2^i}{\Lambda}$.

    We partition $V_2$ into classes as follows. For $k\in\zrnone{K_0}$, let
    \begin{align*}
        Q_k=\set{u\in V_2\mid \abs{\FHs_{G'}(v)\cap \FHs_{G'}(u)}\in\left[\frac1{q_k},\frac2{q_k}\right)}\;, &  & Q_0=\set{u\in V_2\mid \abs{\FHs_{G'}(v)\cap \FHs_{G'}(u)}\geq 1/q_0}\;.
    \end{align*}
    In words, the set $Q_0$ consist of all vertices $u\in V_2$, such that the number of $3$-cycles in $G'$ that contain both $v$ and $u$ is at least $1/q_0=\Lambda$, and for $k>0$, we have that $u\in Q_k$ if and only if the number of $3$-cycles in $G'$ that contain both $v$ and $u$ is at least $1/q_k=\Lambda/2^i$ and less than $2/q_k=\Lambda/{2^{i-1}}$.
    We claim that there exists $k\in\set{0,1,\ldots, K_0}$ such that $\abs{Q_k}\geq 1/p_k$, which we will later prove.
    Given that this holds, we show that for such $k$, we have $v[P_k]\geq \phval[2]$.

    \renewcommand{\DD}{\mathcal{D}}
    Fix $k\in\set{0,1,\ldots, K_0}$ where $\abs{Q_k}\geq 1/p_k$.
    For any non-empty subset $S\subseteq Q_k$, let $\Em{1}(S)$ denote the event that $\set{Q_k\cap V_2[p_k]=S}$.
    That is, $\Em{1}(S)$ denotes the event that during the $P_k$-discovery experiment, the set of vertices that were sampled from $V_2\cap Q_k$ is exactly the set $S$.
    Let $R(S)$ denote the subset of vertices in $V_3$ which participate in a $3$-cycle (in $G'$) that contains $v$ and some additional vertex from $S$.
    Let $\Em{2}(S)$ denote the event that $\set{V_3[q_k]\cap R(S)\nemp}$.
    That is, $\Em{2}(S)$ denotes the event that during the $P_k$-discovery experiment, the set of vertices that were sampled from $V_3\cap R(S)$ is not empty.
    We show that
    \begin{align}
        v[P_k]\geq \sum_{S: \emptyset\subsetneq S\subseteq Q_k} \Pr{\Em{1}(S)\cap \Em{2}(S)} \geq (\eo)^2 \label{eq3:v Pk}\;.
    \end{align}
    For every fixed non-empty $S\subseteq Q_k$, the events $\Em{1}(S)$ and $\Em{2}(S)$ are independent, since $\Em{1}(S)$ addresses sampling vertices from $V_2$, while $\Em{2}(S)$ addresses sampling vertices from $V_3$, where the two samples are independent of each other.
    Also note that $\Em{1}(S)\cap \Em{2}(S)$ is contained in the event that $v$ is $P_k$-discovered, and since the events $\set{\Em{1}(S)}_{S\subseteq Q_k}$ are pairwise disjoint, so are the events $\set{\Em{1}(S)\cap \Em{2}(S)}_{S\subseteq Q_k}$. Therefore, the event that $v$ is $P_k$-discovered, contains the union of the following pairwise disjoint events $\sset{\bigcup_{S: \emptyset\subsetneq S\subseteq Q_k}\Em{1}(S)\cap\Em{2}(S)}$.
    We get
    \begin{align*}
        v[P_k]
         & \geq \Pr{\bigcup_{S: \emptyset\subsetneq S\subseteq Q_k}\Em{1}(S)\cap\Em{2}(S)} \\
         & = \sum_{S: \emptyset\subsetneq S\subseteq Q_k} \Pr{\Em{1}(S)\cap \Em{2}(S)}     \\
         & = \sum_{S: \emptyset\subsetneq S\subseteq Q_k} \Pr{\Em{1}(S)}\Pr{\Em{2}(S)}\;.
    \end{align*}
    To complete the proof, we need to show that
    \begin{enumerate}
        \item $\sum_{S: \emptyset\subsetneq S\subseteq Q_k}\Pr{\Em{1}(S)}\geq \eo$, and
        \item for every non-empty $S\subseteq Q_k$, we have $\Pr{ \Em{2}(S)}\geq \eo$.
    \end{enumerate}
    The first claim follows as
    \begin{align*}
        \sum_{S: \emptyset\subsetneq S\subseteq Q_k}\Pr{\Em{1}(S)}=\Pr{Q_k[p_k]\nemp}=1-(1-p_k)^{\abs{Q_k}}\geq \eo\;,
    \end{align*}
    where the inequalities hold for the following reasons. The first two equalities follows from definition, and the last inequality follows from the assumption that $\abs{Q_k}\geq 1/p_k$.

    The second claim follows as every non-empty subset $S\subseteq Q_k$ satisfies $\abs{R(S)}\geq 1/q_k$.
    To see this, fix some vertex $u\in S$. We have $R(u)\subseteq R(S)$, and $\abs{R(u)}\geq 1/q_k$ because $u\in Q_k$.
    Therefore, for any such $S$, we have
    \begin{align*}
        \Pr{\Em{2}(S)}=\Pr{R(S)[q_k]\nemp}=1-(1-q_k)^{\abs{R(S)}}\geq 1-(1-q_k)^{1/q_k}\geq \eo\;.
    \end{align*}
    This completes the proof of \Cref{eq3:v Pk}.

    \medskip
    To complete the proof of the induction base, it remains to show that our claim indeed holds, i.e., that there exists $k\in\set{0,1,\ldots, K_0}$ such that $\abs{Q_k}\geq 1/p_k$.
    Assume towards a contradiction that this is not the case, that is, that $v$ is a vertex with $t_{G'}(v)\geq \Lambda \cdot \ct[4]$ and for every $k\in\zrn{K_0}$, we have $\abs{Q_k}\cdot p_k<1$. Note that this implies that $Q_0$ is empty.
    Then,
    \begin{align*}
        t_{G'}(v)
         & =   \sum_{u\in[V_2]}\abs{\FHs_{G'}(v)\cap \FHs_{G'}(u)}                            \\
         & <   \sum_{u\in[V_2]}\sum_{1\leq k\leq K_0}\chi_{\set{u\in Q_k}}\cdot \frac{2}{q_k} \\
         & =   \sum_{1\leq k\leq K_0}\sum_{u\in[V_2]}\chi_{\set{u\in Q_k}}\cdot \frac{2}{q_k} \\
         & =   \sum_{1\leq k\leq K_0}\abs{Q_k}\cdot \frac{2}{q_k}
    \end{align*}
    Because we assumed towards a contradiction that for every $k\in\set{0,1,\ldots, K_0}$ we have $\abs{Q_k}<1/p_k$, we get
    \begin{align*}
        \sum_{1\leq k\leq K_0}\abs{Q_k}\cdot \frac{2}{q_k}<
        \sum_{1\leq k\leq K_0}\frac{1}{p_k}\cdot \frac{2}{q_k}=
        \sum_{1\leq k\leq K_0}2\Lambda=2K_0\cdot \Lambda\;.
    \end{align*}
    This is a contradiction, as $2K_0\cdot \Lambda\leq \Lambda \cdot \ct[4]\geq \tvc$.
    This completes the proof of the base case.

    \paragraph*{Induction Step.}
    We are now ready to prove the induction step.
    The induction hypothesis is as follows.
    For an $h$-coloring $\varphi$ and a vertex $v$ with $\tvc\geq \Lambda\cdot \ctz[h-1]$, there exists a vector $P\in \Prd$ such that the probability that $v$ gets $P$-discovered over $\Gpc$ is at least $\phval$.
    We assume the induction hypothesis is true for $h-1$, and prove it for $h$.
    For brevity, we write $G'$ instead of $\Gpc$.

    We define a new set of auxiliary graphs $H_u$ for every $u\in V_h$.
    The graph $H_u$ is an $(h-1)$-partite graph, obtained by removing the vertices and edges of the vertex set $V_h$.
    Then, for every pair of vertices $y\in V_{h-1}$, and $x\in V_1$, add an edge from $y$ to $x$ if and only if $y\to u\to x$ is a directed path in $G'$.
    For any vector $X\in[0,1]^{h-1}$, we say that $X$-discovers $v$ over $H_u$ if after sampling vertices from $H_u$ according to $X$, the vertex $v$ is part of an $(h-1)$-cycle in the sampled subgraph. We denote the probability that $v$ is $X$-discovered over $H_u$ by $v_{H_u}[X]$.

    First, let us provide some insights into how we utilize these auxiliary graphs.
    Let $Y'=(p_1,p_2,\ldots,p_h)$ be some vector in $[0,1]^{h}$, and let $Y$ be its restriction on the first $h-1$ coordinates.
    Informally, for every $u\in V_h$, we reduce the $Y'$-discovery experiment over $G'$, into two independent experiments, one over $V_h$, and one over $H_u$.
    This allows us to show that the probability that $v$ is $Y'$-discovered over $G'$ is at least the probability that $u$ is sampled in the $Y'$-discovery experiment over $G'$, multiplied by the probability that $v$ is $Y$-discovered over $H_u$.
    This can be formalized by the following inequality:
    \begin{align}
        v[Y']\geq v_{H_u}[Y]\cdot p_h\;. \label{eq:claim1}
    \end{align}
    \newcommand{\szy}{S_Y}
    We'll prove a stronger claim, which generalizes \Cref{eq:claim1}.
    The claim is as follows.
    Let $Y'=(p_1,p_2,\ldots,p_h)$ be some vector in $[0,1]^{h}$, let $Y=(p_1,p_2,\ldots, p_{h-1})$, its restriction on the first $h-1$ coordinates, and let $\alpha\geq 0$. Define the set $\szy$ as the set of vertices $u\in V_h$ such that $v_{H_u}[Y]\geq \alpha$.
    Then, $v[Y']\geq \alpha\cdot \Pr{\szy[p_h]\nemp}$.
    We prove this statement in \Cref{prop1:set}.

    ~\\ For the rest of the proof, we find $\szy$ with the ``right'' parameters, such that $v[Y']\geq \phval$.
    In other words, we prove that there exists an index $k\in \zrnone{\log(\Lambda)+1}$, a vector $X\in\Prodd_{h-1}(\frac{2^k}{\Lambda})$, and a set $Q_{k,X}\subseteq V_h$ of size at least $2^k$, such that every vertex $u\in Q_{k,X}$ satisfies $v_{H_u}[X]\geq \phval[h-2]$.
    We prove this statement in \Cref{thm:key}.

    Using \Cref{thm:key,prop1:set}, we can prove the induction step, which we do next.
    Fix some $k$, a vector $X$, and a set $Q_{k,X}$ as specified in \Cref{thm:key}.
    Using $X$, we can define a new vector $X'$ which has $h$ entries, one additional entry compared to $X$, as follows.
    For $i\in [h-1]$, the $i$-th entry of $X'$ is equal the $i$-th entry of $X$.
    The $h$-th entry of $X'$ is equal to $1/2^k$.
    As $X\in\Prodd_{h-1}(\frac{2^k}{\Lambda})$, we have that $X'\in \Prd$. We claim that $v[X']\geq\phval$, which completes the proof of the induction step.
    This follows by \Cref{prop1:set}, as $v[X']\geq \alpha\cdot \Pr{Q_{k,X}[1/2^k]\nemp}$, where $\alpha\geq \phval[h-2]$, by the definition of $Q_{k,X}$. Also, $\Pr{Q_{k,X}[p_h]\nemp}\geq \eo$, as $\abs{Q_{k,X}}\geq 2^k$.
    This completes the proof of the induction step, and therefore the proof of \Cref{prop:discover over Gpi}.
    
\end{proof}
We are left with proving \Cref{prop1:set,thm:key} which contains the statements we used in the induction step. We use the notation from the induction step.
\newcommand{\szy}{S_Y}
\begin{claim}\label{prop1:set}
    Let $Y'=(p_1,p_2,\ldots,p_h)$ be some vector in $[0,1]^{h}$, and let $Y=(p_1,p_2,\ldots, p_{h-1})$.
    Let $\alpha\geq 0$ and let $\szy$ be a subset of vertices from $V_h$, where for every vertex $u\in \szy$, we have $v_{H_u}[Y]\geq \alpha$.
    Then, $v[Y']\geq \alpha\cdot \Pr{\szy[p_h]\nemp}$.
\end{claim}

\begin{proof}[Proof of \Cref{prop1:set}]
    Fix $Y'$ and $\szy$.
    We sample vertices from $V(G')$ according to $Y'$. Let $W$ denote the random variable equal to the set of sampled vertices, and let $W_i\triangleq W\cap V_i$ for $i\in[h]$.

    For any non-empty subset $S\subseteq \szy$, let $\Em{1}(S)$ denote the event that $\set{\szy\cap W_h=S}$.
    That is, $\Em{1}(S)$ denotes the event that during the $Y'$-discovery experiment over $G'$, the set of vertices that were sampled from $V_h\cap \szy$ is exactly the set $S$.

    For every $u\in \szy$, let $\Em{2}(u)$ denote the event that the graph $H_u[W-V_h]$ has an $(h-1)$-cycle, which contains $v$.
    Note that $\Em{2}(u)$ is exactly the event that $v$ is $Y$-discovered over $H_u$. Therefore, we have $\Pr{\Em{2}(u)}=v_{H_u}[Y]\geq \alpha$.
    For any non-empty subset $S\subseteq \szy$, let $\Em{2}(S)$ denote the event that $\set{\bigcup_{u\in S}\Em{2}(u)}$. That is, the event that at for at least one vertex $u\in S$ the random subgraph $H_u[W-V_h]$ has an $(h-1)$-cycle, which contains $v$.
    Next, we show that
    \begin{align}
        v[Y']\geq \sum_{S: \emptyset\subsetneq S\subseteq \szy} \Pr{\Em{1}(S)\cap \Em{2}(S)} \geq \alpha \cdot \Pr{\szy[p_h]\nemp} \label{eq4:vY}\;.
    \end{align}
    We prove several properties of $\Em{1}(S)$ and $\Em{2}(S)$ for every fixed non-empty $S\subseteq Q_k$.
    \begin{enumerate}
        \item $\Em{1}(S)$ and $\Em{2}(S)$ are independent.
        \item $\Em{1}(S)\cap \Em{2}(S)$ is contained in the event that $v$ is $Y'$-discovered.
        \item The events $\set{\Em{1}(S)\cap \Em{2}(S)}_{S\subseteq \szy}$ are pairwise disjoint.
    \end{enumerate}
    The first property follows since $\Em{1}(S)$ addresses sampling vertices from $V_h$, while $\Em{2}(S)$ addresses sampling vertices from $V-V_h$, where the two samples are independent of each other.

    The second property follows from the definitions:
    For any such $S$, the event $\Em{2}(S)$ implies that $W$ contains
    a subset of vertices $(v,v_2,v_{h-1})$ which is an $(h-1)$-path in $H_u$ for some $u\in S$.
    This means that the set of vertices $(v,v_2,v_{h-1})$ is an $(h-1)$-path in $G'$ and that $(v,v_2,v_{h-1},u)$ is an $h$-cycle in $G'$.
    The event $\Em{1}(S)$ implies that $u\in S$ was sampled into $W$, and therefore $v$ is a part of an $h$-cycle in $G'[W]$, which means $v$ was discovered.

    The third property follows because the events $\set{\Em{1}(S)}_{S\subseteq \szy}$ are pairwise disjoint, and therefore so are the events $\set{\Em{1}(S)\cap \Em{2}(S)}_{S\subseteq \szy}$.

    \medskip
    Therefore, the event that $v$ is $Y'$-discovered, contains the union of the following pairwise disjoint events $\sset{\bigcup_{S: \emptyset\subsetneq S\subseteq \szy}\Em{1}(S)\cap\Em{2}(S)}$.
    We get
    \begin{align*}
        v[Y']
         & \geq \Pr{\bigcup_{S: \emptyset\subsetneq S\subseteq \szy}\Em{1}(S)\cap\Em{2}(S)} \\
         & = \sum_{S: \emptyset\subsetneq S\subseteq \szy} \Pr{\Em{1}(S)\cap \Em{2}(S)}     \\
         & = \sum_{S: \emptyset\subsetneq S\subseteq \szy} \Pr{\Em{1}(S)}\Pr{\Em{2}(S)}\;,
    \end{align*}
    which proves \Cref{eq4:vY}.
    To complete the proof of \Cref{prop1:set}, we note that for every non-empty set $S\subseteq \szy$, we have $\Pr{\Em{2}(S)}\geq \alpha$.
    We also note that
    \begin{align*}
        \sum_{S: \emptyset\subsetneq S\subseteq \szy} \Pr{\Em{1}(S)}=
        \Pr{\szy[p_h]\nemp}\;,
    \end{align*}
    by definition.
    This completes the proof of \Cref{prop1:set}.
\end{proof}

\begin{proposition}\label{thm:key}
    There exists an index $k\in \zrnone{\log(\Lambda)+1}$, a vector $X\in\Prodd_{h-1}(\frac{2^k}{\Lambda})$, and a set $Q_{k,X}\subseteq V_h$ of size at least $2^k$, such that every vertex $u\in Q_{k,X}$ satisfies $v_{H_u}[X]\geq \phval[h-2]$.
\end{proposition}

Before proving \Cref{thm:key}, we need a few definitions.
Let $\beta\triangleq 2h\log(n)$.
Recall that $t_{G'}(v)\geq \Lambda\cdot \beta^{(h-1)^2}$, and let $M=\Lambda\cdot \beta^{(h-2)^2}$.
We follow the line of proof of the base case.
Let $K_0=\log(\Lambda) + 1$.
For $0\leq i\leq K_0$, define $p_{i}\triangleq 2^{-i}\;, q_{i}\triangleq 2^i/M$, and partition the vertices of $V_h$ into classes as before, by defining
\begin{align*}
    Q_k=\set{u\in V_h\mid \abs{\FHs_{G'}(v)\cap \FHs_{G'}(u)}\in\left[\frac1{q_k},\frac2{q_k}\right)}\;, &  & Q_0=\set{u\in V_h\mid \abs{\FHs_{G'}(v)\cap \FHs_{G'}(u)}\geq 1/q_0}\;.
\end{align*}
We also need the following claim.

\begin{claim}\label{claim:ythm}
    We claim that at least one of the following holds:
    \begin{enumerate}
        \item The set $Q_0$ is non-empty.
        \item There exists an index $k$ such that $\abs{Q_k}\geq \beta^{h-1}/p_k$.
    \end{enumerate}
\end{claim}

\begin{proof}[Proof of \Cref{thm:key} Using \Cref{claim:ythm}]
    Assume that \Cref{claim:ythm} is true.
    More specifically, that $Q_0$ is not empty, and let $u$ be some vertex in $Q_0$, where $\abs{\FHs_{G'}(v)\cap \FHs_{G'}(u)}\geq 1/q_0=M$.
    We apply the induction hypothesis on $H_u$, which states that there exists at least one vector $P\in \Prodd_{h-1}(\Lambda)$ that $P$-discovers $v$ over $H_u$, with probability at least $\phval[h-2]$.
    In order to use the induction hypothesis on $H_u$, we need to verify that $\FHs_{H_u}(v)\geq \Lambda\cdot\beta^{(h-1)^2}$, which follows
    as $\Lambda\cdot\beta^{(h-1)^2}=M$.
    In words, the number of $(h-1)$-cycles in $H_u$ which contain $v$ is at least $M$.
    We can therefore define the set $Q_{0,P}$, which matches the conditions specified in \Cref{thm:key}.

    Next, we prove \Cref{thm:key} using \Cref{claim:ythm}, assuming
    that there exists $k\in\set{0,1,\ldots,K_0}$ for which $\abs{Q_k}\geq \beta^{h-1}/p_k$, and fix such $k$.
    We apply the induction hypothesis on $H_u$ for every $u\in Q_k$.
    We claim that for every vertex $u\in Q_k$, there exists at least one vector $P\in \Prodd_{h-1}(2^k/\Lambda)$ that $P$-discovers $v$ over $H_u$, with probability at least $\phval[h-2]$. In short, there exists $P\in \Prodd_{h-1}(2^k/\Lambda)$ such that $v_{H_u}[P]\geq \phval[h-2]$.
    In order to use the induction hypothesis on $H_u$, we need to verify that $\FHs_{H_u}(v)\geq \Lambda\cdot\beta^{(h-1)^2}/2^k $, which follows
    by the definition of the set $Q_k$, as for every $u\in Q_k$, we have that $\FHs_{H_u}(v)\geq 1/q_k=\frac{M}{2^k}=\Lambda\cdot\frac{\beta^{(h-1)^2}}{2^k}$.
    In words, the number of $(h-1)$-cycles in $H_u$ which contain $v$ is at least $1/q_k$.

    The final step in the proof is grouping together vertices from $u\in Q_k$ which share the same discovering vector $X$ (over different graphs) to form the set $Q_{k,X}$ which matches the conditions specified in \Cref{thm:key}.

    We use the pigeonhole principle, with the vertices of $Q_k$ as pigeons, and the vectors of $\Prodd_{h-1}(2^k/\Lambda)$ as holes.
    Each vertex in $Q_k$ has at least one vector $P\in \Prodd_{h-1}(2^k/\Lambda)$ such that $v_{H_u}[P]\geq \phval[h-2]$.
    For every vector $X\in \Prodd_{h-1}(2^k/\Lambda)$, let $Q_{k,X}$ denote the set of vertices from $Q_k$ for which $v_{H_u}[X]\geq \phval[h-2]$.
    There are $|\Prodd_{h-1}(2^k/\Lambda)|\leq \beta^{h-1}$ such sets, and each vertex from $Q_k$ is in at least one set, where $Q_k$ has at least $2^k\beta^{h-1}$ vertices.
    Therefore, at least one set contains at least $2^k$ vertices, which completes the proof.

\end{proof}
\begin{proof}[Proof of \Cref{claim:ythm}]
    Assume towards a contradiction that this is not the case, that is, that $v$ is a vertex with $\tvc\geq \Lambda \cdot \beta^{(h-1)^2}$ and
    we have that $Q_0$ is empty, and that for every $k\in\zrnone{K_0}$, we have $\abs{Q_k}\cdot p_k /\beta^{h-1}<1$. Then,
    \begin{align*}
        \tvc
         & =   \sum_{u\in[V_h]}\abs{\FHs_{G'}(v)\cap \FHs_{G'}(u)}                            \\
         & <   \sum_{u\in[V_h]}\sum_{1\leq k\leq K_0}\chi_{\set{u\in Q_k}}\cdot \frac{2}{q_k} \\
         & =   \sum_{1\leq k\leq K_0}\sum_{u\in[V_h]}\chi_{\set{u\in Q_k}}\cdot \frac{2}{q_k} \\
         & =   \sum_{1\leq k\leq K_0}\abs{Q_k}\cdot \frac{2}{q_k}
    \end{align*}
    Because we assumed towards a contradiction that for every $k\in\set{1,2,\ldots, K_0}$ we have $\abs{Q_k}<\frac{\beta^{h-1}}{p_k}$, we get
    \begin{align*}
        \sum_{1\leq k\leq K_0}\abs{Q_k}\cdot \frac{2}{q_k}<
        \sum_{1\leq k\leq K_0}\frac{\beta^{h-1}}{p_k}\cdot \frac{2}{q_k}=
        \sum_{1\leq k\leq K_0}2\cdot \beta^{h-1}\cdot M=
        2K_0\cdot \Lambda\beta^{(h-1)+(h-2)^2}
        \leq 4\Lambda\cdot \beta^{h^2 -3h + 4}\;.
    \end{align*}
    This is a contradiction, as
    $4\Lambda\cdot \beta^{h^2 -3h + 4}
        \leq \Lambda\cdot \beta^{(h-1)^2}\;.$
    The last inequality follows because $h\geq 4$.
\end{proof}

\section{Fine-grained Lower Bound}
\label{sec:lower bound}
A popular hardness hypothesis in fine-grained complexity is that detecting a triangle in an $n$ node graph requires $n^{\omega-o(1)}$ time (see e.g. \cite{vsurvey}). The motivation behind the hypothesis is that the only known subcubic time algorithms for triangle detection use fast matrix multiplication. In fact, it is known that triangle detection and {\em Boolean} matrix multiplication are subcubically equivalent \cite{focsy}: any truly subcubic time algorithm for one problem implies a truly subcubic time algorithm for the other.

Suppose now that one is given a tripartite graph whose three partitions have sizes $n, n^a, n^b$ respectively. Then the fastest known algorithm for finding a triangle in such a graph is again a reduction to fast matrix multiplication, but in this case it is to {\em rectangular} matrix multiplication: $n\times n^a$ by $n^a\times n^b$.

We will shortly see that a faster algorithm would imply (for instance) that $k$-cliques can be detected faster and would be an important breakthrough.
The fastest known algorithm for $k$-clique detection for any constant $k\geq 3$ is a reduction (see \cite{nesetril,vsurvey}) to detecting a triangle in a tripartite graph with partition sizes $n^{\lfloor k/3\rfloor}, n^{\lceil k/3\rceil}$ and $n^{k-\lceil k/3\rceil-\lfloor k/3\rfloor}$.

Thus, the fastest running time for $k$-clique is $\MM{n^{\lfloor k/3\rfloor}, n^{\lceil k/3\rceil}, n^{k-\lceil k/3\rceil-\lfloor k/3\rfloor}}$.
In fact, this running time is conjectured to be optimal in several papers (see e.g. \cite{AbboudBW18,BackursT17,BringmannW17,dalirrooyfard2023listing}). This conjecture is known as the $k$-Clique Hypothesis.

Let's look at $k$ of the form $k=3p-1$ for integer $p\geq 2$. Then $k$-clique detection in $N$-node graphs reduces to triangle detection in a tripartite graph with parts of sizes $N^p,N^{p-1},N^p$. Letting $n=N^p$, the sizes are $n,n^{1-1/p},n$. Thus, if triangle detection in an unbalanced tripartite graph with parts of size $n,n^{1-1/p},n$ can be solved in $\MM{n,n^{1-1/p},n}^{1-\delta}$ time for some $\delta>0$, then the $k$-Clique Hypothesis for $k=3p-1$ would be false.

This motivates the following hypothesis.

\begin{hypothesis}\label{hyp:unblanced}
    Let $G$ be a tripartite graph with vertex set $A\sqcup B\sqcup C$,
    where $\abs{A}=\abs{C} = n$ and $\abs{B}=n^b$ for $b\leq 1$.
    Any randomized algorithm for triangle detection in $G$, in the word-RAM with $O(\log n)$ bit words needs
    $\MM{n,n^b,n}^{1-o(1)}$
    time.
\end{hypothesis}

For $b=1-1/p$ for integer $p\geq 2$ the hypothesis is implied by the $k$-Clique Hypothesis, as mentioned above. When $b\leq \alpha$, the hypothesis is actually true since by definition in this case $\MM{n,n^b,n}\neq n^{2+o(1)}$, and it is not hard to show that to determine whether a tripartite graph with parts of sizes $n,n^b,n$ contains a triangle, any (even randomized) algorithm needs $\Omega(n^2)$ time.

A related hypothesis to Hypothesis \ref{hyp:unblanced} was considered by Bringmann and Carmeli \cite{CarmeliUnbalanced} who proposed that triangle detection in a tripartite graph with partition sizes $n,n^a,n^a$ for $a\leq 1$ requires $\omega(n^{1+a})$ time. Notably, if $\omega=2$, the respective unbalanced matrix multiplication time is $\MM{n,n^a,n^a}=n^{1+a+o(1)}$. The hypothesis of \cite{CarmeliUnbalanced} states that not only is $\MM{n,n^a,n^a}$ time needed, but that the extra $+o(1)$ factor in the exponent is necessary when $\omega=2$.

Other related unbalanced versions of triangle detection appear for instance in \cite{KopelowitzW20} and \cite{monotriangles}, where the tripartite graph is unbalanced not only in the partition sizes but in the number of edges between the pairs of partitions.

We reduce the above unbalanced triangle detection problem for any choice of $b$ to the problem of approximating the number of triangles in an $n$ node graph that contains $0$ triangles if the unbalanced instance had $0$ triangles and at least $n^{1-b}$ triangles if the unbalanced instance had at least one triangle.

Let $t=n^{1-b}$.
We build a new graph $G'$ with vertex set $A\cup B'\cup C$
where
\begin{align*}
    B'\triangleq    & \set{(x,i)\mid x\in B \wedge i\in[t]}                                                    \\
    E(G')\triangleq & E(A,C)\cup \set{(v,(x,i))\mid v\in A\cup C\; \wedge (x,i)\in B' \wedge (v,b)\in E(G)}\;.
\end{align*}
In words, $B'$ contains $t$ copies of each vertex $x\in B$, and $E(G')$ contains all edges between a pair of vertices in $A,C$, and also an edge $(u,(x,i))$ for every $u\in A\cup C$ if $(u,x)\in E(G)$.
Constructing $G'$ takes $\BO{n^2}$ time and $G'$ has $O(n)$ vertices.
\begin{claim}\label{claim:reduce}
    If $G$ has no triangles, then neither does $G'$. If
    $G$ contains a triangle then $G'$ contains $\geq t$ triangles.
\end{claim}
\begin{proof}
    If $G$ has a triangle, it must be of the form $(a,x,c)\in A\times B\times C$. Such a triangle creates $t$ new triangles in $G'$, each of the form
    $(a,(x,i),c)$ in $G'$, showing $G'$ has $t$ triangles.
    For the other direction, note that any triangle $(a,(x,i),c)$ in $G'$, implies that $(a,x,c)$ is a triangle in $G$.
\end{proof}
Any algorithm that can obtain an estimate $t'$ of the number of triangles $t$ in $G$ such that $(1-\eps)t\leq t'\leq (1+\eps)t$ can distinguish between graphs with no triangles and graphs with at least $t$ triangles. Thus we obtain:

\begin{corollary}
    Under Hypothesis 1, any randomized algorithm that can $(1\pm \eps)$-approximate the number of triangles in a graph with $t$ triangles requires $\MM{n,n/t,n}^{1-o(1)}$ time.
\end{corollary}

\subsection{Generalizing for Directed and Odd Cycles}

We now prove our hardness result for $h$-cycles.

\begin{theorem}
    \label{hyp:unblanced k}
    Assume \Cref{hyp:unblanced} and the Word-RAM model of computation with $O(\log n)$ bit words. Let $h\geq 3$ be any constant integer.
    Any randomized algorithm that, when given an $n$ node directed graph $G$, can
    distinguish between $G$ being $C_h$-free and containing $\geq t$ $C_h$s needs
    $\MM{n,n/{t}^{1/(h-2)},n}^{1-o(1)}$
    time. The same result holds for undirected graphs as well whenever $h$ is odd.
\end{theorem}

\begin{proof}
    Let $\ell\triangleq{t}^{1/(h-2)}$.
    Given a graph $G$ with vertex set $A,B,C$ where $\abs{A}=\abs{C}=n$ and $\abs{B}=n/\ell$ in which we want to find a triangle, we build a new graph $G'$ as follows.
    We first define a new set of vertices $B'$ obtained by duplicating each vertex in $B$  $\ell$ times.
    We then duplicate $B'$  $h-2$ times. We denote the copies by $\mathcal{B} \triangleq B_1',\dots,B_{h-2}'$.
    We also write
    \begin{align*}
        \mathcal{B}\triangleq\set{(b,i,j)\mid b\in B\wedge i\in[\ell]\wedge j\in[h-2]}\;.
    \end{align*}
    That is, for any $j\in[h-2]$, we have $B_j\triangleq\set{(b,i,j)\mid b\in B\wedge i\in[\ell]}$ and we have $\mathcal{B}=\cup_j B_j$.
    The vertex set of $G'$ is $A\cup \mathcal{B}\cup C$. We next define the edge set of $G'$.
    \begin{enumerate}
        \item $E_1 \triangleq E(C,A)$.
        \item $E_2 \triangleq \set{(a,(b,i,1))\in A\times B_1'\mid (a,b)\in E}$.
        \item $E_3 \triangleq \set{((b,i,h-2),c)\in B_{h-2}'\times C\mid (b,c)\in E}$.
        \item $E_4 \triangleq \set{((b,i,j),(b,i',j+1))\mid j\in [h-3], i,i'\in [\ell]}$. In other words, all copies of $b$ in $B_j$ point to all copies of $b$ in $B_{j+1}$.
    \end{enumerate}
    We define $E(G')\triangleq E_1\cup E_2\cup E_3\cup E_4$.

    If $G'$ has a $h$-cycle, then this $h$-cycle is of the form

    $$(a\in A)\rightarrow (b,i_1,1)\rightarrow (b,i_2,2)\rightarrow \ldots \rightarrow (b,i_{h-2},h-2)\rightarrow (c\in C)\rightarrow a,$$
    for various choices of $i_j\in [\ell]$.
    This is clear if $G'$ is directed. If $G'$ is undirected and $h$ is odd, this is also the only possible $h$-cycle because $G'$ is layered: the $h$-cycle cannot skip any one of the parts $A,C,B_1,\ldots,B_{h-2}$ because removing any one of these parts makes $G'$ bipartite and so no odd cycle can be contained there.

    On the other hand, a $h$-cycle of the form $(a\in A)\rightarrow (b,i_1,1)\rightarrow (b,i_2,2)\rightarrow \ldots \rightarrow (b,i_{h-2},h-2)\rightarrow (c\in C)\rightarrow a,$ means that $(a,b),(b,c),(c,a)$ are edges of $G$ and hence $G$ has a triangle.

    On the other hand, if $G$ has a triangle $a\in A,b\in B, c\in C$, then for every choice $(i_1,\ldots,i_{h-2})\in [\ell]^{h-2}$, we have a $h$-cycle in $G'$: $a\rightarrow (b,i_1,1)\rightarrow (b,i_2,2)\rightarrow \ldots \rightarrow (b,i_{h-2},kh-2)\rightarrow c\rightarrow a.$
    Thus every triangle in $G$ gives rise to $\ell^{h-2}=t$ $h$-cycles in $G'$.

    Hence, if one can distinguish between $0$ and $t$ $h$-cycles in $G'$ in $O(\MM{n,n/\ell,n}^{1-\delta})$ time for any $\delta>0$, then one can find a triangle in $G$ in asymptotically the same time (as creating $G'$ from $G$ only takes $O(n^2)$ time).

\end{proof}

\bibliographystyle{alpha}
\addcontentsline{toc}{section}{Bibliography}
\bibliography{refs.bib}

\numberwithin{claim}{section}
\begin{appendices}

    \section{Some Matrix Multiplication Tools}\label{sec:appendix}
    \begin{theorem}\label{thm:MM balance is slower}
        Let $\eps>0, a \geq 0, a+\eps \leq 1$. If $\omega(1,1, a)=r$, then $\omega(1,1-\eps, a+$ $\eps) \leq r$.
    \end{theorem}
    \begin{proof}
        Consider multiplying matrices of dimension $n \times n^{1-\eps}$ by matrices of dimension $n^{1-\eps} \times n^{a+\eps}$.
        Define $N=n^{\eps /(1-a)}$. This is possible since $a<1, \eps>0$.
        Further define $q=(1-a-\eps) / \eps \geq 0$. This is also fine because $a+\eps \leq 1$ and $\eps>0$.
        Note that these parameters give that $n=N^{q+1}$.
        Then
        \begin{align*}
            \left\langle n, n^{1-\eps}, n^{a+\eps}\right\rangle=\left\langle N^{q+1}, N^{q+a}, N^{a q+1}\right\rangle=\left\langle N^q, N^q, N^{a q}\right\rangle \otimes\left\langle N, N^a, N\right\rangle .
        \end{align*}
        Thus the rank of $\left\langle n, n^{1-\eps}, n^{a+\eps}\right\rangle$ is at most the product of the ranks of $\left\langle N^q, N^q, N^{a q}\right\rangle$ and $\left\langle N, N^a, N\right\rangle$.

        If $\omega(1, a, 1)=r$, then $\omega(1,1, a)=r$ as well.
        Thus, $\operatorname{Rank}\left(\left\langle N^q, N^q, N^{a q}\right\rangle\right) \leq$ $N^{q r+o(1)}$ and $\operatorname{Rank}\left(\left\langle N, N^a, N\right\rangle\right) \leq N^{r+o(1)}$. Hence, $\operatorname{Rank}\left(\left\langle n, n^{1-\eps}, n^{a+\eps}\right\rangle\right) \leq$ $N^{(q+1) r+o(1)}=n^{r+o(1)}$.
        Thus $\omega(1,1-\eps, a+\eps) \leq r$.
    \end{proof}

    \begin{claim}\label{claim:a2}
        For any $p_1,p_2,p_3\in[0,1]$ we have
        $\MM{n p_1,n p_2,n p_3}\leq \MM{n,n ,n \cdot(p_1\cdot p_2\cdot p_3)}$.
    \end{claim}
    \begin{proof}[Proof of \Cref{claim:a2}]
        We use \Cref{thm:MM balance is slower}, which states that
        $\omega(1,1, a)\geq \omega(1,1-\eps, a+\eps)$, when $a+\eps\leq 1$.
        We will prove that $ \omega(1-(x+y+z),1,1)\geq \omega(1-x,1-y,1-z)$.
        Assume $1\geq x\geq y\geq z\geq 0$.
        We will prove the following transitions:
        \begin{align*}
            \omega(1-(x+y+z),1,1)
            \geq \omega(1-(x+y),1,1-z)
            \geq \omega(1-x,1-y,1-z)\;.
        \end{align*}
        Recall that $\omega(x_1,x_2,x_3)=\omega(x_{\pi(1)},x_{\pi(2)},x_{\pi(3)})$ for any permutation $\pi:[3]\to[3]$.
        \begin{itemize}
            \item The first transition follows by setting $a=1-(x+y+z)$ and $\eps=z$.
                  To apply \Cref{thm:MM balance is slower}, we need to show that $a+\eps\leq 1$, which follows as $a+\eps= 1-x-y\leq 1$.
            \item The second transition follows by setting $a=1-(x+y)$ and $\eps=y$.
                  To apply \Cref{thm:MM balance is slower}, we need to show that $a+\eps\leq 1$, which follows as $a+\eps= 1-x\leq 1$.
        \end{itemize}
        To complete the proof, we simply change the notation. Let $p_1=n^{-x},p_2=n^{-y},p_3=n^{-z}$. Then,
        \begin{align*}
            \MM{n p_1,n p_2,n p_3}=n^{\omega(1-x,1-y,1-z)}\leq n^{\omega(1,1,1-(x+y+z))}= \MM{n,n ,n \cdot(p_1\cdot p_2\cdot p_3)}.
        \end{align*}
    \end{proof}
    \section{The Median Trick}
    \label{app:median}
    Here we prove the median trick as used in the proof of \Cref{lemma3:A4,prop:ffall}.
    \begin{claim}\label{claim:med trick0}
        Let $Y$ be a random variable, and let $[a,b]$ be some interval, where $\Pr{Y\in [a,b]}\geq 2/3$.
        Define $Z$ to be the median of $r=\rr$ independent samples of $Y$. Then, $\Pr{Z\in [a,b]}\geq 1-\frac{1}{n^5}$.
    \end{claim}
    \begin{proof}%
        Let $Y_i$ be the $i$-th sample. We define an indicator random variable $X_i$ to be $1$ if $Y_i\in[a,b]$ for $i\in [r]$.
        Let $X=\sum_{i\in[r]} X_i$.
        We claim that if more than half of the samples fall inside the interval $[a,b]$ then so is the median of all samples. In other words, the rest of the samples do not affect the median.
        Therefore, $\Pr{Z\in[a,b]}\geq \Pr{X>\frac{r}{2}}$.
        We lower bound the probability that $\Pr{X>\frac{r}{2}}$.
        Note that $X$ is the sum of $r$ independent $p$-Bernoulli random variables with $p\geq 2/3$. Then,
        \begin{align*}
            \Pr{X\leq\frac{r}{2}}
            \leq \Pr{\Bin{r}{\frac{2}{3}}\leq \frac{r}{2}}
            = \Pr{\Bin{r}{\frac{2}{3}}\leq \frac{2r}{3}\cdot(1-\frac{1}{4})}\;,
        \end{align*}
        and by the following version of Chernoff's inequality:
        $\Pr{X\leq (1-\delta)\Exp{X}}\leq \exp\brak{-\frac{\delta^2\cdot \Exp{X}}{2}}$, we have
        \begin{align*}
            \Pr{\Bin{r}{\frac{2}{3}}\leq \frac{2r}{3}\cdot(1-\frac{1}{4})}
            \leq \exp\brak{-\frac{r}{48}}
            \leq \exp\brak{-8\log n}
            \leq n^{-5}\;.
        \end{align*}
    \end{proof}
\end{appendices}

\end{document}